\documentclass[runningheads]{llncs}
\usepackage{graphicx}

\usepackage{amssymb, bussproofs, amsmath}
\usepackage[all]{xy}
\usepackage{pgf, tikz, color, changepage, relsize, tcolorbox}
\usetikzlibrary{decorations.pathreplacing}
\usepackage[utf8]{inputenc}
\usepackage{enumerate}
\usepackage{rotating}
\usepackage{appendix,hyperref}

\newcommand{\ldml}{\mathsf{G3Ldm}}
\newcommand{\ldmtl}{\mathsf{G3Tstit}}
\newcommand{\tstitl}{\mathsf{G3Tstit}}

\newcommand{\xstitl}{\mathsf{G3Xstit}}

\newcommand{\ldm}{\mathsf{Ldm}}
\newcommand{\ldmn}{\mathsf{Ldm}}

\newcommand{\ldmt}{\mathsf{Tstit}}
\newcommand{\tstit}{\mathsf{Tstit}}
\newcommand{\xstit}{\mathsf{Xstit}}

\newcommand{\settrefl}{(\mathsf{refl}_{\Box})}
\newcommand{\setteucl}{(\mathsf{eucl}_{\Box})}
\newcommand{\stitrefl}{(\mathsf{refl}_{\agbox{}})}
\newcommand{\stiteucl}{(\mathsf{eucl}_{\agbox{}})}

\newcommand{\ioa}{(\mathsf{IOA})}
\newcommand{\bridge}{(\mathsf{br}_{[i]})}
\newcommand{\gtrans}{(\mathsf{trans}_{\g})}
\newcommand{\gser}{(\mathsf{ser}_{\g})}
\newcommand{\gconn}{(\mathsf{conn}_{\g})}
\newcommand{\hconn}{(\mathsf{conn}_{\h})}
\newcommand{\gconv}{(\mathsf{conv}_{\g})}
\newcommand{\hconv}{(\mathsf{conv}_{\h})}
\newcommand{\ncuh}{(\mathsf{ncuh})}
\newcommand{\agtrefl}{(\mathsf{refl}_{Ag})}
\newcommand{\agteucl}{(\mathsf{eucl}_{Ag})}
\newcommand{\agtd}{(\mathsf{agd})}
\newcommand{\irrtwo}{(\mathsf{irr}_{\g})}
\newcommand{\agr}{(\mathsf{[Ag]})}
\newcommand{\agdiar}{(\mathsf{\lb Ag \rb})}
\newcommand{\reflis}{(\mathsf{refl}_=)}
\newcommand{\euclis}{(\mathsf{eucl_=})}
\newcommand{\subis}{(\mathsf{sub_=})}
\newcommand{\comp}{(\mathsf{comp}_{\g 1})}
\newcommand{\comptwo}{(\mathsf{comp}_{\g 2})}

\newcommand{\diamx}{\langle X \rangle} 
\newcommand{\boxnext}{([\mathsf{X}])}
\newcommand{\diamnext}{(\langle\mathsf{X}\rangle)}
\newcommand{\serx}{(\mathsf{ser_X})}
\newcommand{\detx}{(\mathsf{det_X})}

\newcommand{\stitboxx}{(\mathsf{[A]}^x)}
\newcommand{\stitdiamx}{(\mathsf{\langle A\rangle}^x)}
\newcommand{\ioax}{(\mathsf{IOA_X})}

\newcommand{\cmon}{(\mathsf{C{-}Mon})}
\newcommand{\emptyeff}{({\emptyset}\mathsf{Eff})}
\newcommand{\effempty}{(\mathsf{Eff}{\emptyset})}
\newcommand{\ageff}{(\mathsf{AgEff})}
	\newcommand{\effag}{(\mathsf{EffAg})}

\newcommand{\disr}{(\vee)}
\newcommand{\conr}{(\wedge)}
\newcommand{\settr}{(\Box)}
\newcommand{\settdiar}{(\Diamond)}
\newcommand{\stitr}{(\agbox{})}
\newcommand{\stitdiar}{(\agdia)}

\newcommand{\gr}{(\mathsf{G})}
\newcommand{\fr}{(\mathsf{F})}
\newcommand{\hr}{(\mathsf{H})}
\newcommand{\pr}{(\mathsf{P})}

\newcommand{\id}{(\mathsf{id})}

\newcommand{\R}{\mathcal{R}}

\newcommand{\agdia}{\langle i \rangle}
\newcommand{\agbox}{[i]}

\newcommand{\lb}{\langle}
\newcommand{\rb}{\rangle}
\newcommand{\g}{\mathsf{G}}
\newcommand{\h}{\mathsf{H}}
\newcommand{\f}{\mathsf{F}}
\newcommand{\p}{\mathsf{P}}

\newcommand{\pres}{\mathsf{pres}}
\newcommand{\fut}{\mathsf{fut}}

\EnableBpAbbreviations
\begin{document}
%
%\title{Contribution Title\thanks{Supported by organization x.}}
%
%
\title{Cut-free Calculi and Relational Semantics for Temporal STIT Logics\thanks{This is a pre-print of an article published in Logics in Artificial Intelligence. The final authenticated version is available online at: \url{https://doi.org/10.1007/978-3-030-19570-0_52}.}
}
\author{Kees van Berkel\inst{1} \and Tim Lyon\inst{1}}%\orcidID{0000-0003-3214-0828}}

\institute{$^{1}$Institut f\"ur Logic and Computation, Technische Universit\"at Wien, Austria  \\ \email{\{kees,lyon\}@logic.at}}

%======================================================
%Here are our suggested authorrunning and titlerunning:
%======================================================
\titlerunning{Cut-free Calculi for Temporal STIT Logics}
\authorrunning{K. van Berkel and T. Lyon}

%\author{First Author\inst{1}\orcidID{0000-1111-2222-3333} \and
%Second Author\inst{2,3}\orcidID{1111-2222-3333-4444} \and
%Third Author\inst{3}\orcidID{2222--3333-4444-5555}}
%\institute{Princeton University, Princeton NJ 08544, USA \and
%Springer Heidelberg, Tiergartenstr. 17, 69121 Heidelberg, Germany
\maketitle              % typeset the header of the contribution
\begin{abstract}
%The abstract should briefly summarize the contents of the paper in 15--250 words.
We present cut-free labelled sequent calculi for a central formalism in logics of agency: %of agency:
%a variety of
STIT logics with temporal operators. These include sequent systems for $\ldm$, $\ldmt$ and $\xstit$. % and $\atstit$.
% The latter is introduced as an alternative characterization of the achievement stit through an extension of $\ldmt$.
 All calculi presented possess essential structural properties such as contraction- and cut-admissibility. %Furthermore, t
The labelled calculi $\ldml$ and $\ldmtl$ %and $\atstitl$
 are shown sound and complete relative to irreflexive temporal frames. Additionally, we extend current results by showing
 that also %and obtain that
% the logics
 $\xstit$ %and $\atstit$ 
 can be characterized through relational frames, omitting the use of BT+AC frames. %
 
% $\cdot$ Logics of action}
%\keywords{First keyword  \and Second keyword \and Another keyword.}
\end{abstract}

\section{Introduction}

    Various autonomous machines are developed with the aim of performing particular human tasks. Human acting, however, is inevitably connected to legal and moral decision making--sometimes more than we think. Hence%For that reason
, such machines will eventually be found in difficult scenarios in which normatively acceptable actions must be generated %(i.e. choices made) that are both legally and morally acceptable to us 
\cite{Ger15}. What is more, these decisions can quickly turn into complex (technical) problems \cite{Goo14}. The above stresses the need for formal tools that allow for reasoning about agents, the choices they have, and the actions they are \textit{able} and \textit{allowed} to perform. %What is more, i
 Implementable logics of agency can play an important role in the development of such automated %reasoning
  systems: they can provide explicit proofs that can be checked and which, more importantly, can be understood by humans (\textit{e.g.} \cite{ArkBriBel05}). The present work takes a first step in this direction by providing cut-free sequent calculi for one of the central formalisms of agency: STIT logic with temporal operators. 

% for artificial agents is a long term enterprise, 
% We share the conviction with \cite{Ark05} that machine-implemented logics of agency and permissibility are worth the investigation since they provide explicit reliable proofs of permissible and obligatory actions whose decision making process must be, as apposed to for example machine learning, transparent and, in retrospect, understandable by humans (e.g. \cite{Ark05}). 
 
The logic of STIT, which is an acronym for `Seeing To It That', is a prominent modal framework for the formal analysis of multi-agent interaction and reasoning about choices.\footnote{For an introduction to STIT logic and a historical overview we refer to \cite{BelPer90,BelPerXu01,Hor01}.} In short, STIT logics contain modal formulae of the form $[i]\phi$, capturing the notion that ``the agent $i$ sees to it that the state of affairs $\phi$ is brought about". STIT logic knows many fruitful extensions %variations
 and its recent application to legal theory, deontic reasoning, and epistemics shows that issues of agency are essentially tied to \textit{temporal aspects} of choice: for example,  consider issues in  %in the following topics %exploit
	% temporal constructions are advocated: 
	legal responsibility \cite{LorSar15}; social commitment  \cite{Lor13}; knowledge-based obligations \cite{Bro11}; agent-bound instrumentality \cite{BerPas18}; and actions as events \cite{Xu12}.
 %\footnote{The following agency logics advocate %exploit
	% temporal constructions: Legal responsibility \cite{LorSar15}; Social commitment  \cite{Lor13}; Knowledge-based obligations \cite{Bro11}; Agent-bound instrumentality \cite{BerPas18}; Action events \cite{Xu12}.} %The above provides strong reasons for investigating temporal extensions of STIT. 

Unfortunately, nearly all available proof systems for STIT logics are Hilbert-style systems, which are known to be cumbersome for proof search and not suitable %very useful 
 for proving metalogical properties of the intended formalisms. % formalized logics \cite{}. 
% for proof checking. Unfortunately, most STIT logics have a Hilbert-style proof system (cf. \cite{xu1,xu2,xu3}) and these are 
 To this purpose, a renowned alternative proof framework %\textcolor{red}{to the Hilbert system} 
 is Gentzen's \textit{sequent calculus} \cite{Gen35}. It allows one to construct proofs that decompose the formulae to be proven in a stepwise manner; making it an effective tool for proof search and a good candidate for automated deduction procedures. However, this framework is not strong %powerful
 enough to design cut-free analytic calculi for many modal logics of interest \cite{Neg05}; including STIT logic. In this work, we will treat several STIT logics through %employ 
 a more expressive extension of this %e Gentzen 
 formalism: Labelled Sequent Calculi \cite{Neg05,Vig00}.
 
The aim of the present paper is to provide labelled  calculi for several central \textit{temporal} STIT logics: $\ldm$, $\ldmt$ and $\xstit$. To our knowledge, there have only been three attempts to capture STIT logics in alternative proof systems: in \cite{ArkBriBel05} a natural deduction system for a deontic STIT logic is proposed and in \cite{WanOlk18,Wan06} tableaux systems for multi-agent deliberative STIT logics are presented. 

On the one hand, the novelty of the present contribution compared to previous works, 
is that all presented calculi (i) possess useful proof-theoretic properties such as contraction- and cut-admissibility and (ii) are modular and extend to several \textit{temporal} STIT-logics, including both temporal operators and inherently temporal STIT-operators (in a multi-agent, as well as a group setting). In doing so, we answer an open question in \cite{Wan06} regarding the construction of a rule-based proof system for temporal extensions of $\ldm$. On the other hand, the investigation of STIT has been with an essential focus on its intuitive semantics: branching time structures, extended with histories as paths and agential choice-functions (BT+AC-frames). Recent work \cite{BalHerTro08,HerSch08,Lor13}, however, shows that the basic atemporal STIT logic $\ldm$ and its temporal extension $\ldmt$ are characterizable through simpler relational frames. The current work extends these results by showing that also the logic $\xstit$ can be semantically characterized without using BT+AC structures.

In section \ref{LDM_Section} we will introduce the base logic $\ldm$ and its corresponding labelled calculus. Thereafter, in section \ref{TSTIT_Section}, we provide a cut-free calculus for the temporal STIT logic $\ldmt$, introduced in \cite{Lor13}, which exploits a temporal irreflexivity rule based on %(\textit{cf.} 
 \cite{Gab81}%,GabHodRey94}
%)
. Last, in section \ref{XSTIT_Section}, we provide a labelled calculus for the inherently temporal STIT logic $\xstit$ from \cite{Bro11,Bro11b}. Here we show that the \textit{independence of agents} principle of STIT logic can be captured using \textit{systems of rules} from \cite{Neg16}. We conclude and highlight some envisaged future work in section \ref{Conclusion}.

\section{The Logic $\ldmn$}\label{LDM_Section}

   \subsection{Axioms and Relational Semantics for $\ldmn$}

The basic STIT logic $\ldm$ offers a framework for reasoning about individual agents realizing propositions via the %ir 
choices available to them at particular moments in time.  
%Due to the simplicity of $\ldm$, the STIT logic has served as a base for more expressive STIT logics that make use of deontic, group, and temporal operators (See \cite{Mur04,HerSch08,Lor13}). NOTE: I think this is not really needed, when can add them in the introduction?
%
%For example, in \cite{Mur14} $\ldm$ is extended with deontic operators to provide an axiomatization for a utilitarian ethical theory. In \cite{HerSch08}, the complexity of $\ldm$ is discussed and an extension with group operators is shown undecidable. In \cite{Lor13}, $\ldm$ is extended with temporal operators and the logic is applied in formalizing the normative notion of social commitment.
%
%The logic EXPLAIN APPLICATIONS AND EXTENSIONS.
%1. Murakami Util. STIT \cite{Mur04}
%2. Herzig and Swchaztentruber Group STIT \cite{HerSch08}
%3. Lorini TSTIT \cite{Lor13}
%
 In the semantics of $\ldm$, each \textit{moment} can be formalized as an equivalence class of \textit{worlds}, where each world sits in a linear chain (referred to as a \emph{history}) extending to the future and (possibly to) the past. Therefore, each world contained in a particular moment can be thought of as an alternative state of affairs that evolves along a different timeline. Moreover, for each agent, moments are further partitioned into equivalence classes, where each class represents a possible choice available to the agent for realizing %materializing
 a set of potential outcomes%situations NOTE: I thought, maybe situations is too ambiguous w.r.t. worlds and states
. Hence, if a proposition $\phi$ holds true in every world of a particular choice for an agent $i$, then we claim that ``$i$ sees to it that $\phi$" (written formally as $[i]\phi$) at each world of that choice; \textit{i.e.} %this follows from the fact that
 $i$'s committal to the choice ensures %the truth of
  $\phi$ regardless of which world in the choice set is actual.

The above STIT operator $[i]$ is referred to as the \emph{Chellas}-STIT %operator 
(i.e. \emph{cstit}) \cite{BelPerXu01}. It is often distinguished from the \emph{deliberative} STIT (i.e. \emph{dstit}) which consists of %builds off of  %the interpretation of
 \textit{cstit} together with a negative condition: we say that ``agent $i$ deliberatively sees to it that $\phi$'' (written formally as $[i]^{d}$) %at a world $w$ in a moment $m$ 
 when (i) ``$i$ sees to it that $\phi$'' %at $w$
 and (ii) ``$\phi$ is currently not settled true'' \cite{BelHor95,Hor01}. % there exists another world $w'$ at $m$ where $\phi$ does not hold. NOTE: I would not provide a `model' interpretation of dstit here. 
% Due to the second condition (ii), it is easy to see that the realization of $\phi$ \textit{depends} on the choice made by the agent, \textit{i.e.} $\phi$ might not have been case had the agent chosen to act differently.
 The second condition ensures that the realization of $\phi$ \textit{depends} on the choice made by the agent; \textit{i.e.} $\phi$ might not have been case had the agent chosen to act differently. By making use of the \emph{settledness} operator $\Box$, which is prefixed to a formula when the formula holds true at every world in a moment, \textit{cstit} and \textit{dstit} become inter-definable: namely, $[i]^{d}\phi$ \textit{iff} $[i] \phi \wedge \neg \Box \phi$. As an example of a STIT formula, the formula $\Diamond [i]^d\phi$ must be interpreted as follows: at the current moment, agent $i$ has a possible choice available that allows $i$ to see to it that $\phi$ is guaranteed, and there is an alternative choice present to $i$ that does not guarantee $\phi$. In this paper, we introduce $\Box$ and $[i]$ as primitive and take $[i]^d$ as defined. 
%  Due to the fact that the operators $\Box$ and $[i]$ are in the language of $\ldm$, we present the logic with only the cstit operator $[i]$ as primitive and take the dstit operator $[i]^{d}$ as defined.
%\\
%\textcolor{red}{check:one more line of an extra stit explanation}

In this section, we make all of the aforementioned notions formally precise and provide a relational semantics for $\ldm$ along with a corresponding cut-free labelled calculus. In section \ref{TSTIT_Section}, we will extend $\ldm$ with temporal operators, obtaining the logic $\tstit$. Since both logics rely on the same semantics, we introduce their languages and semantics simultaneously, avoiding unnecessary repetition. Lastly, in what follows we give all formulae of the associated logics in negation normal form. This reduces the number of rules in the associated calculi and offers a simpler presentation of the proof theory.
 The languages for $\ldm$ and $\tstit$ are given below:
 
  % for the logic.
%We will now make all of the aforementioned notions formally precise and provide a relational semantics for $\ldm$ along with a cut-free labelled calculus for the logic. As mentioned previously, $\ldm$ gives rise to many fruitful extensions, and so, in section \ref{TSTIT_Section} we will extend the logic $\ldm$ with temporal operators, thus obtaining Lorini's logic $\tstit$ \cite{Lor13}; we will also provide a cut-free calculus for the logic. NOTE: I commented this one out, since it was or less a repitition of the introduction. 

\begin{definition}[The Languages $\mathcal{L}_{\ldm}$ and $\mathcal{L}_{\ldmt}$]\label{ldmtlanguage} Let $Ag = \{1,2,...,n\}$ be a finite set of agent labels and let $Var =\{p_1,p_2,p_3...\}$ be a countable set of propositional variables. The language $\mathcal{L}_{\ldm}$ is given by the following BNF grammar: 
{\small $$\phi ::= p \ | \ \overline{p} \ | \ \phi \wedge \phi \ | \ \phi\vee \phi \ | \ \Box \phi \ | \ \Diamond \phi \ | \ [i] \phi \ | \ \lb i \rb \phi$$} %where $i\in Ag$ and $p \in Var$
%where $i\in Ag$ and $p \in Var$. 
The language $\mathcal{L}_{\ldmt}$ is defined accordingly: %as an extension of the above:
 %, the language $\mathcal{L}_{\ldm}$ is defined as follows:  
{\small $$\phi ::= p \ | \ \overline{p} \ | \ \phi \wedge \phi \ | \ \phi\vee \phi \ | \ \Box \phi \ | \ \Diamond \phi \ | \ [i] \phi \ | \ \lb i \rb \phi \ | \ [Ag] \phi \ | \ \lb Ag \rb \phi \ | \ \g \phi \ | \ \f \phi \ | \ \h \phi \ | \ \p \phi$$}
where $i\in Ag$ and $p \in Var$.
\end{definition}

\noindent The language $\mathcal{L}_{\tstit}$ extends $\mathcal{L}_{\ldm}$ through the incorporation of the tense modalities $\g$, $\f$, $\h$, and $\p$ and the modalities $[Ag]$ and $\lb Ag \rb$ for the grand coalition $Ag$ of agents. $\g$ and $\f$ are duals and read, respectively, as `always will be in the future' and `somewhere in the future'. $\h$ are $\p$ are also dual and are interpreted, respectively, as `always has been in the past' and `somewhere in the past' (cf. \cite{Lor13,Pri67}). The operator 
$[Ag]$ captures the notion that `the grand coalition of agents sees to it that'. Note that the negation of a formula $\phi$, written $\overline{\phi}$, is obtained in the usual way by replacing each operator with its dual, each positive propositional atom $p$ with its negation $\overline{p}$, and each negative propositional atom $\overline{p}$ with its positive version $p$. We may therefore %introduce the following definitions: $\phi \rightarrow \psi = \overline{\phi} \vee \psi$, $\phi \leftrightarrow \psi = \phi \rightarrow \psi \wedge \psi \rightarrow \phi$, $\top = p \vee \overline{p}$, and $\bot = p \wedge \overline{p}$. 
 define $\phi \rightarrow \psi$ as $\overline{\phi} \vee \psi$, $\phi \leftrightarrow \psi$ as $\phi \rightarrow \psi \wedge \psi \rightarrow \phi$, $\top$ as $p \vee \overline{p}$, and $\bot$ as $p \wedge \overline{p}$. We will use these abbreviations throughout the paper. 
 
 %We will continue to make use of these abbreviations throughout the remainder of the paper. %$\ldm$ is axiomatized accordingly \cite{?}: 

 At present, we are principally %specifically
  interested in $\ldm$ and temporal frames: in particular, since $\ldmt$ will be introduced as the temporal extension of $\ldm$ and, more generally, because the logic of STIT has an implicit temporal intuition underlying choice-making (\textit{cf.} original branching-time frames employed for $\ldm$ \cite{BelPerXu01,BelHor95,Hor01}). We will prove that $\ldm$ is strongly complete with respect to these more elaborate \textit{irreflexive} Temporal Kripke STIT frames.
%At present, we are specifically interested in $\ldm$ and temporal frames, since $\ldmt$ will be introduced as the temporal extension of $\ldm$. We will prove that $\ldm$ is strongly complete with respect to the more elaborate Temporal Kripke STIT frames. %We use the frames defined in \cite{Lor13}:

\begin{definition}[Relational $\ldmt$ Frames and Models \cite{Lor13}]\label{models_tkstit} Let $\R_{\alpha}(w) := \{v\in W | (w,v)\in R_{\alpha}\}$ for $\alpha \in \{\Box, Ag, \g, \h \} \cup Ag$. A relational \emph{Temporal STIT frame ($\ldmt$-frame)} is defined as a tuple $F = (W, \R_{\Box}, \{\R_{i} | i \in Ag\}, \R_{Ag}, \R_{\g}, \R_{\h})$ where $W$ is a non-empty set of worlds $w,v,u...$ and: %\neq \emptyset$ and:

\begin{itemize}

%\item $W$ is a non-empty set of worlds;% $w,v,u...$;

\item For all $i\in Ag$, $\R_{\Box}$, $\R_{i}$, $\R_{Ag} \subseteq W\times W$ are equivalence relations where:

%\begin{itemize}

\item[{\rm \textbf{(C1)}}] For each $i$, $\R_{i} \subseteq \R_{\Box}$;

\item[{\rm \textbf{(C2)}}] For all $u_{1},...,u_{n} \in W$, if $\R_{\Box}u_{i}u_{j}$ for all $1 \leq i,j \leq n$, then $\bigcap_{i} \R_{i}(u_{i}) \neq \emptyset$;

\item[{\rm \textbf{(C3)}}] For all $w \in W$, $\R_{Ag}(w) = \bigcap_{i \in Ag} \R_{i}(w)$;

%\end{itemize}

\item $\R_{\g}\subseteq W\times W$ is a transitive and serial binary relation % defined between worlds of $W$,
and $\R_{\h}$ is the converse of $\R_{\g}$, and the following conditions hold:

%\begin{itemize}

\item[{\rm \textbf{(C4)}}] For all $w, u, v \in W$, if $\R_{\g}wu$ and $\R_{\g}wv$, then $\R_{\g}uv$, $u = v$, or $\R_{\g}vu$;

\item[{\rm \textbf{(C5)}}] For all $w, u, v \in W$, if $\R_{\h}wu$ and $\R_{\h}wv$, then $\R_{\h}uv$, $u = v$, or $\R_{\h}vu$;

\item[{\rm \textbf{(C6)}}] $\R_{\g} \circ \R_{\Box} \subseteq \R_{Ag} \circ \R_{\g}$; (Relation composition $\circ$ is defined as usual.)

\item[{\rm \textbf{(C7)}}] For all $w,u \in W$, if $u \in \R_{\Box}(w)$, then $u \not\in \R_{\g}(w)$;
%\end{itemize}
\end{itemize}
A $\ldmt$-model is defined as a tuple $M = (F,V)$ where $F$ is a $\ldmt$-frame and $V$ is a valuation function assigning propositional variables to subsets of $W$; that is, $V{:}\ Var \mapsto \mathcal{P}(W)$. 
\end{definition}

The property expressed in C2 corresponds to the familiar \textit{independence of agents} principle of STIT logic, which states that if it is currently possible for each distinct agent to make a certain choice, then it is possible for all such choices to be made simultaneously. Condition C6 captures the STIT principle of \textit{no choice between undivided histories}, which %The principle %axiom
 ensures that if two %histories (
 time-lines %)
  remain %are still
   undivided at some %the next %pass through the same 
  future moment, then no agent can currently make a choice realizing one time-line %history and not 
  without the other. %Nevertheless, 
  (This %is another %property is an
% essential condition imposed in STIT logics, which 
principle is %known to be 
 inexpressible in the atemporal language of the %mere
  base logic $\ldm$.) For a philosophical discussion of these principles see \cite{BelPerXu01}. Last, condition C7 ensures that the temporal frames under consideration are \textit{irreflexive}, which means that the future is a strict future (excluding the present). For a discussion of the other frame properties we refer to \cite{Lor13}.

\begin{definition}[Semantics for $\mathcal{L}_{\ldm}$ and $\mathcal{L}_{\ldmt}$%Satisfaction, Global Truth, Validity
]\label{Semantics_ldm_tstit} Let $M$ % = (W, \R_{\Box}, \{\R_{i} | i \in Ag\}, \R_{Ag},$ $\R_{\g}, \R_{\h}, V)$ 
 be a $\ldmt$-model and let $w$ be a world in its domain $W$. The \emph{satisfaction} of a formula $\phi$ on $M$ at $w$ is inductively defined as follows (in clauses 1-14 we omit explicit mention of $M$): %(written $M, w \models \phi$) as follows: NOTE: I think we don't need to explain notation. People should know this. 
\begin{small}

\begin{multicols}{2}
\begin{itemize}
%\vspace{-0.3cm}
\item[1.] $ w \models p$ iff $w \in V(p)$

\item[2.] $ w \models \overline{p}$ iff $w \not\in V(p)$

\item[3.] $ w \models \phi \wedge \psi$ iff $ w \models \phi$ and $ w \models \psi$

\item[4.] $ w \models \phi \vee \psi$ iff $ w \models \phi$ or $ w \models \psi$

\item[5.] $ w \models \Box \phi$ iff $\forall u \in \R_{\Box}(w)$, $ u \models \phi$

\item[6.] $ w \models \Diamond \phi$ iff $\exists u \in \R_{\Box}(w)$, $ u \models \phi$

\item[7.] $ w \models [i] \phi$ iff $\forall u \in \R_{i}(w)$, $ u \models \phi$
\end{itemize}
\columnbreak
\begin{itemize}
\item[8.] $ w \models \lb i \rb \phi$ iff $\exists u \in \R_{i}(w)$, $ u \models \phi$

\item[9.] $ w \models [Ag] \phi$ iff $\forall u \in \R_{Ag}(w)$, $ u \models \phi$

\item[10.] $ w \models \lb Ag \rb \phi$ iff $\exists u \in \R_{Ag}(w)$, $ u \models \phi$

\item[11.] $ w \models \g \phi$ iff $\forall u \in \R_{\g}(w)$, $ u \models \phi$

\item[12.] $ w \models \f \phi$ iff $\exists u \in \R_{\g}(w)$, $ u \models \phi$
\item[13.] $ w \models \h \phi$ iff $\forall u \in \R_{\h}(w)$, $ u \models \phi$

\item[14.] $ w \models \p \phi$ iff $\exists u \in \R_{\h}(w)$, $ u \models \phi$
\end{itemize}
%\vspace{-0.3cm}
\end{multicols}

\end{small}
\noindent A formula $\phi$ is \emph{globally true} on %a model
 $M$ (\textit{i.e.} $M {\models} \phi$) iff it is satisfied at every world $w$ in the domain $W$ of $M$. A formula $\phi$ is \emph{valid} (\textit{i.e.} ${\models} \phi$) iff it is globally true on every $\ldmt$-model.

%Last, we say that a formula $\phi$ is a \emph{semantic consequence} of a set of formulae $\Theta$ (written $\Theta {\models} \phi$) iff for any TK-STIT model $M$ and world $w$ in $M$, if $M,w {\models} \psi$ for all $\psi \in \Theta$, then $M,w {\models} \phi$.\textcolor{purple}{Also I think we might silently omit this definition of semantic consequence, if we decide to put the completeness proof to the appendix, we might want to `recall' this definition there.}

\end{definition}

\begin{definition}[The Logic $\ldm$~\cite{BelPerXu01}] The Hilbert system of $\ldm$ consists of the following axioms and inference rules:

\begin{small}
\begin{center}
\begin{tabular}{c @{\hskip 1em} c @{\hskip 1em} c @{\hskip 1em} c}

$\phi \rightarrow (\psi \rightarrow \phi)$

&

$(\overline{\psi} \rightarrow \overline{\phi}) \rightarrow (\phi \rightarrow \psi)$

&

$(\phi \rightarrow (\psi \rightarrow \chi)) \rightarrow ((\phi \rightarrow \psi) \rightarrow (\phi \rightarrow \chi))$

\end{tabular}

\ \\

\begin{tabular}{c @{\hskip 1em} c @{\hskip 1em} c @{\hskip 1em} c @{\hskip 1em} c @{\hskip 1em} c}

$\Box \phi \rightarrow \phi$

&

$\Diamond \phi \rightarrow \Box \Diamond \phi$

&

$\Box (\phi \rightarrow \psi) \rightarrow (\Box \phi \rightarrow \Box \psi)$

&

$\agbox{} \phi \rightarrow \phi$

&

$\lb i \rb \phi \rightarrow \agbox{} \lb i \rb \phi$

\end{tabular}
\end{center}

\begin{center}
\begin{tabular}{c @{\hskip 1em} c @{\hskip 1em} c @{\hskip 1em} c @{\hskip 1em} c}

$\Box \phi \vee \Diamond \overline{\phi}$

&

$[i] \phi \vee \lb i \rb \overline{\phi}$

&

$\bigwedge_{i \in Ag} \Diamond [i] \phi_{i} \rightarrow \Diamond ( \bigwedge_{i \in Ag}[i] \phi_{i})$

\end{tabular}
\end{center}

\begin{center}
\begin{tabular}{c @{\hskip 1em} c @{\hskip 1em} c @{\hskip 1em} c @{\hskip 1em} c}

$\agbox{} (\phi \rightarrow \psi) \rightarrow (\agbox{} \phi \rightarrow \agbox{} \psi)$

&

$\Box \phi \rightarrow [i] \phi$

&

\AxiomC{$\phi$}
\UnaryInfC{$\Box \phi$}
\DisplayProof

&

\AxiomC{$\phi$}
\AxiomC{$\phi \rightarrow \psi$}
\BinaryInfC{$\psi$}
\DisplayProof

\end{tabular}
\end{center}
\end{small}
A derivation of $\phi$ in $\ldm$ from a set of premises $\Theta$, is written as $\Theta \vdash_{\ldm} \phi$. When $\Theta$ is the empty set, we refer to $\phi$ as a \emph{theorem} and write $\vdash_{\ldm} \phi$.
\end{definition}
%A derivation of $\phi$ in $\mathsf{X} \in \{\ldm, \tstit, \xstit, \atstit \}$ from a set of premises $\Theta$, is written as $\Theta \vdash_{\mathsf{X}} \phi$. Additionally, when $\Theta$ is the emptyset, we refer to $\phi$ as a \emph{theorem} and write $\vdash_{\mathsf{X}} \phi$.
The axiomatization contains duality-axioms $\Box \phi \vee \Diamond \overline{\phi}$ and $[i] \phi \vee \lb i \rb \overline{\phi}$ which ensure the usual interaction between the box and diamond modalities. Furthermore, the axiom $\bigwedge_{i \in Ag} \Diamond [i] \phi_{i} \rightarrow \Diamond ( \bigwedge_{i \in Ag}[i] \phi_{i})$ is the \textit{independence of agents} (\emph{IOA}) axiom. %

\begin{theorem}[Soundness~\cite{Lor13}]
For any formula $\phi$, if $ \vdash_{\ldm} \phi$, then $ {\models} \phi$.
\end{theorem}

%\begin{proof} Follows from theorem \ref{Soundness_Completeness_T-STIT} (see section \ref{TSTIT_Section}).
%\end{proof}

\noindent Observe that all axioms of $\ldm$ are within the Sahlqvist class. Therefore, we know that $\ldm$ is already strongly complete relative to the simpler class of frames defined by the first-order properties corresponding to its axioms \cite{BlaRijVen01} (\textit{cf.} \cite{BalHerTro08,HerSch08} for alternative %contains an explicit 
 completeness proofs of $\ldm$ relative to this class of 
 relational frames). As mentioned previously, we are interested in $\ldm$ relative to the more involved \textit{temporal} frames. The usual canonical model construction from \cite{BlaRijVen01} cannot be %directly
 applied to obtain completeness %immediately 
 of $\ldm$ in relation to $\ldmt$-frames. This follows from the fact that the axioms of $\ldm$ do not impose any temporal structure on the canonical model of $\ldm$, and hence, we are not ensured that the resulting model qualifies % satisfies the conditions necessary to qualify
 as a $\ldmt$-model. Theorem \ref{compl_ldm} is therefore proved via an alternative canonical model construction, which can be found in appendix \ref{Completeness_Proof_ldm}. %NOTE: I would say it is not to us to explicitly state that it is non-trivial; hence, the proof of completeness is non-trivial. Since the details of the $\ldm$ completeness proof are tedious, we merely quote the completeness theorem below and save the details of the rigorous proof for the appendix (\ref{Completeness_Proof_ldm}).

%Since our main focus is sequent calculi for STIT logics, we omit the lengthy completeness proof here, and refer the interested reader to the appendix (available at \url{http://arxiv.org/abs/1902.06632}). 

\begin{theorem}[Completeness]\label{compl_ldm} Any consistent set $\Sigma \subset \mathcal{L}_{\ldm}$ is satisfiable.

\end{theorem}

%\begin{proof} \textcolor{blue}{Follows from lemmas \ref{Lindenbaum}, \ref{Canonical_is_TTKSTIT_Model}, \ref{Truth_Lemma} (found in appendix \ref{Completeness_Proof_ldm}).}

%\end{proof}

\subsection{A Cut-free Labelled Calculus for $\ldm$}\label{cutfreeldm}

We now provide a cut-free labelled calculus for $\ldm$, which can be seen as a simplification of the tableaux calculus in \cite{Wan06}. Labelled sequents $\Gamma$ are defined through the following BNF grammar: $$\Gamma ::= x:\phi \ | \ \Gamma, \Gamma \ | \ \R_{\alpha}xy, \Gamma$$
where $x$ is from a countable set of labels $L = \{ x, y, z, ... \}$, $\alpha \in \{\Box\} \cup Ag$, and $\phi \in \mathcal{L}_{\ldm}$. Note that commas are used equivocally in the interpretation of a labelled sequent: representing (i) a conjunction when occurring between relational atoms, (ii) a disjunction when occurring between labelled formulae, and (iii) an implication when binding the multiset of relational atoms to the multiset of labelled formulae, which comprise a sequent. Last, we use the notation $\vdash_{G3\mathsf{X}} x:\phi$ (for $\mathsf{X} \in \{\ldm, \tstit, \xstit\}$) to denote here and later that the labelled formula $x:\phi$ is derivable in the calculus $\mathsf{G3X}$. 

The first order correspondents of all $\ldm$ axioms are \textit{geometric axioms}:
%The first order correspondents of all axioms of $\ldm$ are geometric axioms:
 that is, axioms of the form $\forall x_{1} ... x_{n} ((\phi_{1} \wedge ... \wedge \phi_{m}) \rightarrow \exists y_{1}...y_{k} (\psi_{1} \vee ... \vee \psi_{l}))$ where each $\phi_{i}$ %($1 {\leq} i {\leq} m$)
  is atomic and does not contain free occurrences of $y_{j}$ (for 
 $1 \leq j \leq k$), and each $\psi_{i}$ %($1 {\leq} h {\leq} l$)
  is a conjunction $\chi_{1} \wedge ... \wedge \chi_{r}$ of atomic formulae. The calculus $\ldml$ is obtained by transforming all such correspondents into rules; \textit{i.e.} \textit{geometric rules}. (For further discussion on extracting rules from axioms, we refer to \cite{Neg05,Neg16}.) Last, since our formulae are in negation normal form, we provide a one-sided version of the calculi introduced in \cite{Neg05}. This allows for a simpler formalism with fewer rules, but which is equivalent in expressivity.

\begin{definition}[The Calculus $\ldml$]\label{sequentldml} %The labelled calculus $\ldml$ consists of the following rules: 

\begin{small}
\begin{center}
\begin{tabular}{c c c c}

\AxiomC{ }
\RightLabel{$\id$}
\UnaryInfC{$\Gamma, w:p, w:\overline{p} $}
\DisplayProof

&

\AxiomC{$\Gamma, w: \phi$}
\AxiomC{$\Gamma, w: \psi$}
\RightLabel{$\conr$}
\BinaryInfC{$\Gamma, w: \phi \wedge \psi$}
\DisplayProof

&

\AxiomC{$\Gamma, w: \phi, w : \psi$}
\RightLabel{$\disr$}
\UnaryInfC{$\Gamma, w: \phi \vee \psi$}
\DisplayProof

&

\AxiomC{$\Gamma, \R_{\Box}wv, v: \phi$}
\RightLabel{$\settr^{*}$}
\UnaryInfC{$\Gamma, w: \Box \phi$}
\DisplayProof
\end{tabular}
\end{center}

\begin{center}
\begin{tabular}{c @{\hskip 1 em} c @{\hskip 1 em} c}

\AxiomC{$\Gamma, \R_{\Box}wu, w: \Diamond \phi, u: \phi$}
\RightLabel{$\settdiar$}
\UnaryInfC{$\Gamma, \R_{\Box}wu, w: \Diamond \phi$}
\DisplayProof

&

\AxiomC{$\Gamma, \R_{i}wv, v: \phi$}
\RightLabel{$\stitr^{*}$}
\UnaryInfC{$\Gamma, w: [i] \phi$}
\DisplayProof

&

\AxiomC{$\Gamma, \R_{i}wu, w: \lb i \rb \phi, u: \phi$}
\RightLabel{$\stitdiar$}
\UnaryInfC{$\Gamma, \R_{i}wu, w: \lb i \rb \phi$}
\DisplayProof

\end{tabular}
\end{center}

\begin{center}
\begin{tabular}{c c c}

\AxiomC{$\R_{\Box}ww, \Gamma$}
\RightLabel{$\settrefl$}
\UnaryInfC{$\Gamma$}
\DisplayProof

&

\AxiomC{$\R_{i}ww, \Gamma$}
\RightLabel{$\stitrefl$}
\UnaryInfC{$\Gamma$}
\DisplayProof

&

\AxiomC{$\R_{\Box}wu_{1}, ..., \R_{\Box}wu_{n}, \R_{1}u_{1}v, ..., \R_{n}u_{n}v, \Gamma$}
\RightLabel{$\ioa^{*}$}
\UnaryInfC{$\R_{\Box}wu_{1}, ..., \R_{\Box}wu_{n},\Gamma$}
\DisplayProof

\end{tabular}
\end{center}

\begin{center}
\begin{tabular}{c c c}

\AxiomC{$\R_{\Box}wu, \R_{\Box}wv, \R_{\Box}uv, \Gamma$}
\RightLabel{$\setteucl$}
\UnaryInfC{$\R_{\Box}wu, \R_{\Box}wv, \Gamma$}
\DisplayProof

&

\AxiomC{$\R_{\Box}wu, \R_{i}wu, \Gamma$}
\RightLabel{$\bridge$}
\UnaryInfC{$\R_{i}wu, \Gamma$}
\DisplayProof

&

\AxiomC{$\R_{i}wu, \R_{i}wv, \R_{i}uv, \Gamma$}
\RightLabel{$\stiteucl$}
\UnaryInfC{$\R_{i}wu, \R_{i}wv, \Gamma$}
\DisplayProof

\end{tabular}
\end{center}
The `$\ast$' on the labels $\settr$, $\stitr$, and $\ioa$ indicates an eigenvariable condition for this rule: \textit{i.e.} the label $v$ occurring in the premise of the rule cannot occur in the conclusion.
\end{small}
\end{definition}

%\subsection{Soundness}

%    \input{Section-1.2.tex} NOTE: I JUST COPIED THIS SECTION IntO THE MAIN FILE AND DELETED THe OTher
%\noindent In order to obtain soundness and completeness for $\ldml$ relative to the class of $\ldmt$-frames we need to spell out how to formally interpret a labelled sequent relative to a given model. Again, to avoid unnecessary repetitions, we provide the semantics uniformly for all labelled sequent languages appearing in this paper (for an extensive elaboration on the topic see \ref{Neg???}).
%
%Just as the language of the logic $\tstit$ extends the language of the logic $\ldm$, the language of the calculus for $\tstit$ extends the language of the calculus for $\ldm$. We therefore also provide the definition of an interpretation of a labelled sequent for the calculus $\tstitl$ for $\tstit$ since this notion is relevant for soundness and completeness in the next section.

\noindent %Note that 
The rule $\id$ is an initial sequent and the rules $\conr$, $\disr$, $\settr$, $\settdiar$, $\stitr$ and $\stitdiar$ allow us to decompose connectives. Furthermore, as indicated by the relational atoms, the rules $\settrefl, \stitrefl, \setteucl,\stiteucl,\bridge$ capture the behavior of the corresponding modal operators, and the rule $\ioa$ secures independence of agents in $\ldml$. In order to establish the intended soundness and completeness results, we need to formally interpret a labelled sequent relative to a given model. For the sake of brevity,
%To save ourselves from repeating similar definitions %Kees: I believe that this is not an appropriate way of saying this, although it is what we are doing..
we provide the semantics uniformly for all labelled sequent languages appearing in this paper: 

%PREVIOUS: To give the reader an idea of how such an interpretation is defined, 

%\noindent We now introduce the formal interpretation of a labelled sequent relative to a given model, which is relevant for establishing soundness and completeness results in the coming sections. To save ourselves from repeating similar definitions, we provide the semantics uniformly for all labelled sequent languages appearing %in this paper.

\begin{definition}[Interpretation, Satisfiability, Validity] %Let $\mathsf{X}\in \{\ldmt,\xstit,\atstit\}$. Let $M$ be an $\mathsf{X}$-model with domain $W$, $L$ the set of labels used in the labelled sequent language of $\mathsf{G3X}$, $\Gamma$ a sequent in $\mathsf{G3X}$ and let $I$ be an \emph{interpretation function} of $L$ on $M$ that maps labels to worlds; \textit{i.e.} $I{:} \ L \mapsto W$. 
%
%Let $M$ be a $\mathsf{X}$-model with $I$ an interpretation of $L$ on $M$. 
%A sequent $\Gamma$ is \emph{satisfied} in $M$ with $I$ iff for all relational atoms $\R_{\alpha} xy$ and equalities $x {=} y$ in $\Gamma$, if $\R_{\alpha} x^{I}y^{I}$ holds in $M$, then there must exist some $z : \phi$ in $\Gamma$ such that $M, z^{I} {\models} \phi$ (where $\R_{\alpha} \in \{\R_{\Box}, \R_{i}, \R_{Ag}, \R_{\g}, \breve{\R}_{\g}, \R_{\h}\}$ for $\mathsf{X}{=}\ldmt$, $\R_{\alpha} \in \{\R_{\Box}, \R_{X}, \R_{A}\}$--for all $A \subseteq Ag$--for $\mathsf{X}{=}\xstit$ and $\R_{\alpha} \in \{\R_{\Box}, \R_{i}, \R_{Ag}, \R_{\g}, \breve{\R}_{\g}, \R_{\h},$ $\breve{\R}_{\h}, \R_{\T}, \R_{\Y}\}$, for $\mathsf{X}{=}\atstit$).\footnote{We interpret $\breve{\R}_{\g}$ to be the complement of the relation $\R_{\g}$ and $\breve{\R}_{\h}$ to be the complement of $\R_{\h}$.}
%
%A sequent $\Gamma$ is \emph{valid} iff it is satisfiable in every model $M$ with every $I$ of $L$ on $M$.
%
Let $\mathsf{X}\in \{\ldm,\ldmt,$ $\xstit\}$. Let $M$ be a model for $\mathsf{X}$ with domain $W$, $L$ the set of labels used in the labelled sequent language of $\mathsf{G3X}$, $\Gamma$ a sequent in $\mathsf{G3X}$ and let $\R_{\alpha}$ be a relation of $M$. %\footnote{
(We have $\R_{\alpha} \in \{\R_{\Box}, \R_{i}\}$ for $\mathsf{X} = \ldm$, $\R_{\alpha} \in \{\R_{\Box}, \R_{i}, \R_{Ag}, \R_{\g}, \breve{\R}_{\g}, \R_{\h}\}$ for $\mathsf{X} = \ldmt$, and $\R_{\alpha} \in \{\R_{\Box}, \R_{X}, \R_{A}\}$, for all $A \subseteq Ag$, when $\mathsf{X} = \xstit$. We take $\breve{\R}_{\g}$ as the complement of the relation $\R_{\g}$.% and $\breve{\R}_{\h}$ to be the complement of $\R_{\h}$.
%}
) Last, let $I$ be an \emph{interpretation function} of $L$ on $M$ that maps labels to worlds; \textit{i.e.} $I{:} \ L \mapsto W$. We say that, 
\begin{itemize}
%Let $M$ be a $\mathsf{X}$-model with $I$ an interpretation of $L$ on $M$. 
\item[] a sequent $\Gamma$ is \emph{satisfied} in $M$ with $I$ iff for all relational atoms $\R_{\alpha} xy$ and equalities $x {=} y$ in $\Gamma$, if $\R_{\alpha} x^{I}y^{I}$ holds in $M$, then there must exist some $z : \phi$ in $\Gamma$ such that $M, z^{I} {\models} \phi$.
\end{itemize}
A sequent $\Gamma$ is \emph{valid} iff it is satisfiable in any model $M$ with any $I$ of $L$ on $M$.
%Let $M$ be an $\mathsf{X}$-model with $\mathsf{X}\in \{\ldmt,\xstit,\atstit\}$ with domain $W$, and let $L$ be the set of labels used in the labelled sequent language. An \emph{interpretation $I$ of $L$ in $M$} is a function $I : L \mapsto W$.

%Let $M$ be a $\mathsf{X}$-model with $I$ an interpretation of $L$ on $M$. We say that a sequent $\Gamma$ is \emph{satisfied} by $M$ with $I$ iff for all relational atoms $\R xy$ with $\R \in \{\R_{\Box}, \R_{i}, \R_{Ag}, \R_{\g}, \breve{\R}_{\g}, \R_{\h}\}$ ($\R \in \{\R_{\Box}, \R_{X}, \R_{A}\}$ with $A \subseteq Ag$, $\R \in \{\R_{\Box}, \R_{i}, \R_{Ag}, \R_{\g}, \breve{\R}_{\g}, \R_{\h},$ $\breve{\R}_{\h}, \R_{\T}, \R_{\Y}\}$, resp.) and equalities $x = y$ in $\Gamma$, if $\R x^{I}y^{I}$ holds in $M$, then there must exist some $z : \phi$ in $\Gamma$ such that $M, z^{I} {\models} \phi$\footnote{We interpret $\breve{\R}_{\g}$ to be the complement of the relation $\R_{\g}$ and $\breve{\R}_{\h}$ to be the complement of $\R_{\h}$.}.

%We say that a sequent $\Gamma$ is \emph{valid} iff it is satisfiable in every model $M$ and every interpretation $I$ of $L$ in $M$.

\end{definition}

\begin{theorem}[Soundness]\label{Soundness_G3LDM}
Every sequent derivable in $\ldml$ is valid.
\end{theorem}

\begin{proof} By induction on the height of the given $\ldml$ derivation. For initial sequents of the form $\Gamma, w{:}p, w {:} \overline{p}$ the claim is clear. %straightforward.
 The inductive step is argued by showing that each inference rule preserves validity (cf. theorem 5.3 in \cite{Neg09}). %  (The proof is similar to the proof of theorem 5.3 in \cite{Neg09}.)

\end{proof}

%\subsection{Completeness}

%    \input{Section-1.1.tex} NOTE: I JUST COPIED THIS SECTION IntO THE MAIN FILE
    
    %The derivation of each axiom and inference rule of $\ldm$ (with the exception of the IOA axiom) is a straightforward exercise (See \cite{Negvon01} and \cite{Neg05} for details). Therefore, we only present the derivation of the IOA axiom (for simplicity, we consider the two agent instance of the axiom; the general version is similar). 

\begin{lemma}\label{sequent_ldm} For all $\phi\in \mathcal{L}_{\ldm}$, if $\vdash_{\ldm} \phi$, then $\vdash_{\ldml} x : \phi$.
\end{lemma}

%\begin{proof} The derivation of each axiom and inference rule of $\ldm$, except for the IOA-axiom, is  straightforward (See \cite{Negvon01,Neg05}). The derivation of IOA can be found in appendix \ref{LDM_IOA}.
%\end{proof}

\begin{proof} The derivation of each axiom and inference rule of $\ldm$, except for the IOA-axiom, is straightforward (See \cite{Neg05,Negvon01}). For readability, we only present the derivation of the IOA-axiom for two agents; the general case is similar:\\

\begin{small}
\begin{tabular}{@{\hskip -4em} c}
\AxiomC{$\R_{1}vu,\R_{1}yv, \R_{1}yu, ..., y : \lb 1 \rb \overline{\phi}_{1}, u : \overline{\phi}_{1},  u : \phi_{1}$}
\UnaryInfC{$\R_{1}vu,\R_{1}yv, \R_{i}yu, ..., y : \lb 1 \rb \overline{\phi}_{1}, u : \phi_{1}$}
%\UnaryInfC{$\R_{1}vu,\R_{1}yv, ..., y : \lb 1 \rb \overline{\phi}_{1}, u : \phi_{1}$}
\dashedLine \UnaryInfC{$\R_{1}vu,\R_{1}yv, ..., y : \lb 1 \rb \overline{\phi}_{1}, u : \phi_{1}$}
\UnaryInfC{$\R_{1}yv, ..., y : \lb 1 \rb \overline{\phi}_{1}, v: [1] \phi_{1}$}

\AxiomC{$\R_{2}vu,\R_{2}zv, \R_{i}zw, ..., z : \lb 2 \rb \overline{\phi}_{2}, w : \overline{\phi}_{2},  w : \phi_{2}$}
\UnaryInfC{$\R_{2}vw,\R_{2}zv, \R_{2}zw, ..., z : \lb 2 \rb \overline{\phi}_{2}, w : \phi_{2}$}
%\UnaryInfC{$\R_{2}vw,\R_{2}yv, ..., z : \lb 2 \rb \overline{\phi}_{2}, w : \phi_{2}$}
\dashedLine \UnaryInfC{$\R_{2}vw,\R_{2}zv, ..., z : \lb 2 \rb \overline{\phi}_{2}, w : \phi_{2}$}
\UnaryInfC{$\R_{2}zv, ..., z : \lb 2 \rb \overline{\phi}_{2}, v: [2] \phi_{2}$}

\BinaryInfC{$\R_{1}yv, \R_{2}zv,\R_{\Box}xy, \R_{\Box}yv, \R_{\Box}xv, \R_{\Box}xz, y : \lb 1 \rb \overline{\phi}_{1}, z: \lb 2 \rb \overline{\phi}_{2}, x: \Diamond ([1] \phi_{1} \wedge [2] \phi_{2}), v: [1] \phi_{1} \wedge [2] \phi_{2}$}

\UnaryInfC{$\R_{1}yv, \R_{2}zv,\R_{\Box}xy, \R_{\Box}yv, \R_{\Box}xv, \R_{\Box}xz, y : \lb 1 \rb \overline{\phi}_{1}, z: \lb 2 \rb \overline{\phi}_{2}, x: \Diamond ([1] \phi_{1} \wedge [2] \phi_{2})$}

\dashedLine \UnaryInfC{$\R_{1}yv, \R_{2}zv,\R_{\Box}xy, \R_{\Box}yv, \R_{\Box}xz, y : \lb 1 \rb \overline{\phi}_{1}, z: \lb 2 \rb \overline{\phi}_{2}, x: \Diamond ([1] \phi_{1} \wedge [2] \phi_{2})$}

\UnaryInfC{$\R_{1}yv, \R_{2}zv,\R_{\Box}xy, \R_{\Box}xz, y : \lb 1 \rb \overline{\phi}_{1}, z: \lb 2 \rb \overline{\phi}_{2}, x: \Diamond ([1] \phi_{1} \wedge [2] \phi_{2})$}
\UnaryInfC{$\R_{\Box}xy, \R_{\Box}xz, y : \lb 1 \rb \overline{\phi}_{1}, z: \lb 2 \rb \overline{\phi}_{2}, x: \Diamond ([1] \phi_{1} \wedge [2] \phi_{2})$}
\UnaryInfC{$x : \Box \lb 1 \rb \overline{\phi}_{1}, x: \Box \lb 2 \rb \overline{\phi}_{2}, x: \Diamond ([1] \phi_{1} \wedge [2] \phi_{2})$}
\UnaryInfC{$x : \Box \lb 1 \rb \overline{\phi}_{1} \vee \Box \lb 2 \rb \overline{\phi}_{2} \vee \Diamond ([1] \phi_{1} \wedge [2] \phi_{2})$}
\DisplayProof
\end{tabular}
\end{small}

\vspace{0.3cm}
\noindent The dashed lines in the above proof indicate the use of transitivity rules, which are derivable from the $\stitrefl$, $\stiteucl$, $\settrefl$, and $\setteucl$ rules (see \cite{Neg05}).

\end{proof}

\begin{theorem}[Completeness]\label{Compeleteness_G3ldm} For all $\phi \in \mathcal{L}_{\ldm}$, if $\models \phi$, then $\vdash_{\ldml} x: \phi$.
\end{theorem}

\begin{proof} Follows from theorem \ref{compl_ldm} and lemma \ref{sequent_ldm}.

\end{proof}
Due to the fact that 
%Since %
 %$\ldml$ and all other calculi %in this paper, %
all labelled sequent calculi given in this paper
% to be presented
  fit within the scheme presented in \cite{Neg05,Neg16}, we obtain the subsequent theorem specifying their proof-theoretic properties:

%THIS IS A EMERGENCY SHORT RENDITION:
%\textcolor{red}{All calculi given in this paper fit within the scheme presented in \cite{Neg05,Neg16} and thus we obtain the subsequent theorem specifying their proof-theoretic properties:}

%Due to the modularity of the labelled formalism within which the calculus is given, we obtain cut-free calculi for extensions of $\ldm$ whose frame conditions are specified by generalized geometric axioms (See \cite{Neg05,Neg16}).

\begin{theorem}\label{Properties_of_Calc} Each calculus $\mathsf{G3X}$ with $\mathsf{X} \in \{\ldm, \ldmt, \xstit \}$ has the following properties:

\begin{enumerate}

\item  All sequents of the form $\Gamma, x : \phi, x: \overline{\phi}$ are derivable in $\mathsf{G3X}$ with $\phi$ in the language $\mathcal{L}_{\mathsf{X}}$;

\item All inference rules of $\mathsf{G3X}$ are height-preserving invertible;

%\item $\mathsf{G3X}$ has the sub-term property (\textit{i.e.} any term found in a minimal derivation occurs in its endsequent or is an eigenvariable);

\item Weakening, contraction, and variable-substitution are height-preserving admissible;

\item Cut is admissible.

\end{enumerate}
\end{theorem}

%\textcolor{purple}{sub-term property. And excellent structural properties according to Negri hehe. ``Our calculi, although not satisfying a full subformula property, enjoy a subterm property: all terms in minimal derivations are terms found in the endsequent. This property, together with height-preserving admissibility of contraction, makes our calculi suitable for proof search.'' [p.508] negri modal logic)}

\begin{proof} See \cite{Neg05} and \cite{Neg16} for details.
\end{proof}

In order to maintain the admissibility of contraction, our calculi must satisfy the \emph{closure condition} \cite{Neg05,Neg16}. That is, %Hence, we stipulate that 
 the calculi $\ldml,\ldmtl$ and $\xstitl$ adhere to the following condition: For any \textit{generalized geometric rule} in which a substitution of variables produces a duplication of relational atoms or equalities active in the rule, the instance of the rule with such duplicates contracted is added to the calculus. Since variable substitutions can only bring about a finite number of rule instances possessing duplications, the closure condition adds at most finitely many rules and is hence unproblematic. (Generalized geometric rules extend the class of geometric rules % discussed previously. In short, the former can be 
 and can be extracted from generalized geometric axioms. In short, these are formulae of the form $GA_{n} = \forall x_{1} ... x_{n} ((\phi_{1} \wedge ... \wedge \phi_{m}) \rightarrow (\exists y_{1} \bigwedge GA_{k_{1}} \vee ... \vee \exists y_{m} \bigwedge GA_{k_{m}}))$, %where each $\phi_{i}$ is atomic and does not contain free occurrences of $y_{j}$ ($1 \leq j \leq m$) and 
 where each $\bigwedge GA_{k_{j}}$ (for $0 {\leq} k_{1}, \cdots\!, k_{m} {<} n$) stands for a conjunction of generalized geometric axioms, inductively %defined up to the $k_{j}^{th}$ step of the inductive 
 constructed up to $k_j$-depth with the base case $GA_{0}$ being a geometric axiom. For a formal treatment of these axioms and rules see \cite{Neg16}.)

%In order to maintain the admissibility of contraction, our calculi must satisfy the \emph{closure condition} (see \cite{Neg05}). Hence, we stipulate that all calculi $\mathsf{X} \in \{\ldm,\ldmt,\xstit\}$ adhere to the following condition: For any (generalized) geometric rule in which a substitution of variables produces a duplication of relational atoms or equalities active in the rule, the instance of the rule with such duplicates contracted is added to the calculus. Since variable substitutions can only bring about a finite number of rule instances possessing duplications, the closure condition adds at most finitely many rules and is hence unproblematic.
%
%\textcolor{blue}{Note that generalized geometric rules extend the class of geometric rules discussed previously. As explained in \cite{Neg16}, such rules can be extracted from generalized geometric formulae, which are formulae of the form}
%
%\begin{center}
%$GA_{n} = \forall x_{1} ... x_{n} ((\phi_{1} \wedge ... \wedge \phi_{m}) \rightarrow \exists y_{1} \bigwedge GA_{k_{1}} \vee ... \vee \exists y_{m} \bigwedge GA_{k_{m}}))$
%\end{center}
%
%\noindent
%\textcolor{blue}{where each $\phi_{i}$ is atomic and does not contain free occurrences of $y_{j}$ (for $1 \leq j \leq m$), $0 \leq k_{1}, \cdots, k_{m} < n$, $\bigwedge GA_{k_{j}}$ stands for a conjunction of geometric formulae (defined up to the $k_{j}^{th}$ step of the inductive construction procedure), and in the base case $GA_{0}$ is a geometric formula.}

%\subsection{Proof-theoretic Properties}

%    \input{Section-1.3.tex}

\section{The Logic $\ldmt$}\label{TSTIT_Section}

    \subsection{Axiomatization for $\ldmt$}

The logic $\tstit$ extends the logic $\ldm$ through the incorporation of %the
 tense modalities % $\g$, $\f$, $\h$, and $\p$ 
 and the modality %ies $[Ag]$ and $\lb Ag \rb$ 
 for the grand coalition %$Ag$ 
 of agents (see definition \ref{ldmtlanguage}). %$\g$ and $\f$ are duals and read, respectively, as `always will be in the future' and `somewhere in the future'. $\h$ are $\p$ are also dual and are interpreted, respectively, as `always has been in the past' and `somewhere in the past' (cf. \cite{Pri67,Lor13}). 
 This additional expressivity allows for the application of $\ldmt$ in settings where one wishes to reason about the joint action of all agents, or %where one wishes to reason about 
 the consequences of choices over time. The logic was originally proposed in \cite{Lor13} as a Hilbert system, in this section we provide a corresponding cut-free calculus. %In fact, the logic has been applied in \cite{Lor13} to formalize the notion of \emph{social commitment}, where one agent is committed to another agent to see to it that some proposition is realized in the future. 

%In the present section, we introduce the logic as well as provide a cut-free calculus for it.

\begin{definition}[The Logic $\ldmt$ \cite{Lor13}]\label{TSTIT_Def} The Hilbert system for the logic $\ldmt$ is defined as the logic $\ldm$ \emph{extended} with the following axioms and inference rules:
\begin{small}
\begin{center}
\begin{tabular}{c @{\hskip 1.5em} c @{\hskip 1.5em} c @{\hskip 1.5em} c}
$[Ag] \phi \rightarrow \phi$
&
$\lb Ag \rb \phi \rightarrow [Ag] \lb Ag \rb \phi$
&
$\bigwedge_{1 \leq i \leq n} [i] \phi_{i} \rightarrow [Ag] \bigwedge_{1 \leq i \leq n} \phi_{i}$
&
$\phi \rightarrow \g \p \phi$
\end{tabular}
\end{center}
\begin{center}
\begin{tabular}{c @{\hskip 1.5em} c @{\hskip 1.5em} c @{\hskip 1.5em} c @{\hskip 1.5em} c}
$\phi \rightarrow \h \f \phi$
&
$\g \phi \rightarrow \f \phi$
&
$\f \f \phi \rightarrow \f \phi$
&
$\f \p \phi \rightarrow \p \phi \vee \phi \vee \f \phi$
&
$\p \f \phi \rightarrow \p \phi \vee \phi \vee \f \phi$
\end{tabular}
\end{center}
\begin{center}
\begin{tabular}{c @{\hskip 1.2em} c @{\hskip 1.2em} c @{\hskip 1.4em} c}
$\g \phi \vee \f \overline{\phi}$
&
$\h \phi \vee \p \overline{\phi}$
&
$[Ag] \phi \vee \lb Ag \rb \overline{\phi}$
&
$\alpha (\phi \rightarrow \psi)\rightarrow (\alpha \phi\rightarrow \alpha\psi)$ for $\alpha\in\{\g,\h,[Ag]\}$
\end{tabular}
\end{center}
\begin{center}
\begin{tabular}{c @{\hskip 2em} c @{\hskip 2em} c @{\hskip 2em} c %@{\hskip 2em} c
}
$\f \Diamond \phi \rightarrow \lb Ag \rb \f \phi$
&
\AxiomC{$\phi$}
\UnaryInfC{$\g \phi$}
\DisplayProof
&
\AxiomC{$\phi$}
\UnaryInfC{$\h \phi$}
\DisplayProof
&
\AxiomC{$(\Box \neg p \wedge \Box (\g p \wedge \h p)) \rightarrow \phi$}
\RightLabel{\text{ with $p \not\in \phi$}}
\UnaryInfC{$\phi$}
\DisplayProof
\end{tabular}
\end{center}
\end{small}
A derivation of $\phi$ in $\tstit$ from a % (possibly empty)
 set of premises $\Theta$, is written as $\Theta \vdash_{\tstit} \phi$. When $\Theta$ is the empty set, we refer to $\phi$ as a \emph{theorem} and write $\vdash_{\ldmt} \phi$.
\end{definition}

\noindent Note that the axiom $\f \Diamond \phi \rightarrow \lb Ag \rb \f \phi$ characterizes the \emph{no choice between undivided histories} property (%i.e. C6,
definition \ref{models_tkstit}, C6). Furthermore, the last inference rule, a variation of Gabbay's irreflexivity rule \cite{Gab81}, characterizes the %frame
 property of $\R_{\g}$-irreflexivity (%i.e. C7, 
definition \ref{models_tkstit}, C7). %is needed to ensure that the present does not occur in the future. 
 For a discussion of all axioms and %inference
  rules see \cite{Lor13}.

%The central idea is that since both histories come together in a future moment, the two histories will evolve through the same moments until that time, and therefore, the realization of one history must realize the other.

%Recall that $\ldmt$-models and their semantics are given in definition \ref{models_tkstit} and \ref{Semantics_ldm_tstit}, respectively)

\begin{theorem}[Soundness and Completeness~\cite{Lor13}]\label{Soundness_Completeness_T-STIT} For any formula $\phi \in \mathcal{L}_{\tstit}$, $\vdash_{\tstit} \phi$ iff $ \models \phi$.
\end{theorem}

\subsection{A Cut-free Labelled Calculus for $\tstit$}

Let $L=\{ x, y, z, ... \}$ be a countable set of labels. The language of $\ldmtl$ is defined as follows: 
%\vspace{-0.2cm}
%\begin{center}
$$\Gamma ::= x:\phi \ | \ \Gamma, \Gamma \ | \ \R_{\alpha} xy, \Gamma$$
%\end{center} 
%\vspace{-0.2cm}
where $x\in L$, $\phi\in\mathcal{L}_{\ldmt}$, and $\R_{\alpha} \in \{\R_{\Box}, \R_{i}, \R_{Ag}, \R_{\g}, \breve{\R}_{\g}, \R_{\h}\}$. On the basis of this language, we construct the calculus $\tstitl$ as an extension of $\ldml$. 

%\newpage

\begin{definition}[The Calculus $\ldmtl$]\label{sequentldmtl} The labelled calculus $\ldmtl$ consists of all the rules of $\ldml$ extended with the following set of rules:
\begin{small}
\begin{center}
\begin{tabular}{c c c c}

\AxiomC{$\R_{\h}wu, \R_{\g}uw, \Gamma$}
\RightLabel{$\hconv$}
\UnaryInfC{$\R_{\h}wu, \Gamma$}
\DisplayProof

&

\AxiomC{$ \Gamma, \R_{\h}wu, w: \p \phi, u : \phi$}
\RightLabel{$\pr$}
\UnaryInfC{$\Gamma, \R_{\h}wu, w: \p \phi$}
\DisplayProof

&

\AxiomC{ }
\RightLabel{$\comp$}
\UnaryInfC{$\R_{\g}wu, \breve{\R}_{\g}wu, \Gamma$}
\DisplayProof

\end{tabular}
\end{center}

\begin{center}
\begin{tabular}{c c c c}
\AxiomC{$\Gamma, \R_{\g}wv, v: \phi$}
\RightLabel{$\gr^{*}$}
\UnaryInfC{$\Gamma, w: \g \phi$}
\DisplayProof

&

\AxiomC{$\Gamma, \R_{\g}wu, w: \f \phi, u : \phi$}
\RightLabel{$\fr$}
\UnaryInfC{$\Gamma, \R_{\g}wu, w: \f \phi$}
\DisplayProof

&

\AxiomC{$ \R_{\g}wu, \R_{\h}uw, \Gamma$}
\RightLabel{$\gconv$}
\UnaryInfC{$ \R_{\g}wu, \Gamma$}
\DisplayProof

&

\end{tabular}
\end{center}

\begin{center}
\begin{tabular}{c c c c}

\AxiomC{$ \Gamma, \R_{Ag}wu, w: \lb Ag \rb \phi, u : \phi$}
\RightLabel{$\agdiar$}
\UnaryInfC{$\Gamma, \R_{Ag}wu, w: \lb Ag \rb \phi$}
\DisplayProof

&

\AxiomC{$\R_{Ag}ww, \Gamma$}
\RightLabel{$\agtrefl$}
\UnaryInfC{$\Gamma$}
\DisplayProof

&

\AxiomC{$w = w, \Gamma$}
\RightLabel{$(\mathsf{refl_{=}})$}
\UnaryInfC{$\Gamma$}
\DisplayProof

\end{tabular}
\end{center}

\begin{center}
\begin{tabular}{c}

\AxiomC{$\R_{\g}uv, \R_{\g}wu, \R_{\g}wv, \Gamma$}
\AxiomC{$u = v, \R_{\g}wu, \R_{\g}wv, \Gamma$}
\AxiomC{$\R_{\g}vu, \R_{\g}wu, \R_{\g}wv, \Gamma$}
\RightLabel{$\gconn$}
\TrinaryInfC{$\R_{\g}wu, \R_{\g}wv, \Gamma$}
\DisplayProof

\end{tabular}
\end{center}

\begin{center}
\begin{tabular}{c}

\AxiomC{$\R_{\h}uv, \R_{\h}wu, \R_{\h}wv, \Gamma$}
\AxiomC{$u = v, \R_{\h}wu, \R_{\h}wv, \Gamma$}
\AxiomC{$\R_{\h}vu, \R_{\h}wu, \R_{\h}wv, \Gamma$}
\RightLabel{$\hconn$}
\TrinaryInfC{$\R_{\h}wu, \R_{\h}wv, \Gamma$}
\DisplayProof
\end{tabular}
\end{center}

\begin{center}
\begin{tabular}{c c}
\AxiomC{$ \R_{\g}wu, \R_{\Box}uz, \R_{Ag}wv, \R_{\g}vz, \Gamma$}
\RightLabel{$\ncuh^{*}$}
\UnaryInfC{$ \R_{\g}wu, \R_{\Box}uz, \Gamma$}
\DisplayProof

&

\AxiomC{$\R_{\g}wu, \Gamma$}
\AxiomC{$\breve{\R}_{\g}wu, \Gamma$}
\RightLabel{$\comptwo$}
\BinaryInfC{$\Gamma$}
\DisplayProof

\end{tabular}
\end{center}

\begin{center}
\begin{tabular}{c c c}

\AxiomC{$\R_{\g}wu, \R_{\g}uv, \R_{\g}wv, \Gamma$}
\RightLabel{$\gtrans$}
\UnaryInfC{$\R_{\g}wu, \R_{\g}uv, \Gamma$}
\DisplayProof

&

\AxiomC{$\R_{Ag}wu, \R_{i}wu, \Gamma$}
\RightLabel{$\agtd$}
\UnaryInfC{$\R_{Ag}wu, \Gamma$}
\DisplayProof

&

\AxiomC{$ \Gamma, \R_{\h}wv, v: \phi$}
\RightLabel{$\hr^{*}$}
\UnaryInfC{$ \Gamma, w: \h \phi$}
\DisplayProof

\end{tabular}
\end{center}

\begin{center}
\begin{tabular}{c c c}

\AxiomC{$ \R_{\g}wv, \Gamma$}
\RightLabel{$\gser^{*}$}
\UnaryInfC{$ \Gamma$}
\DisplayProof

&

\AxiomC{$\R_{Ag}wu, \R_{Ag}wv, \R_{Ag}uv, \Gamma$}
\RightLabel{$\agteucl$}
\UnaryInfC{$\R_{Ag}wu, \R_{Ag}wv, \Gamma$}
\DisplayProof

&

\AxiomC{$\R_{\Box}wu, \breve{\R}_{\g}wu, \Gamma$}
\RightLabel{$\irrtwo$}
\UnaryInfC{$\R_{\Box}wu, \Gamma$}
\DisplayProof
\end{tabular}
\end{center}

\begin{center}
\begin{tabular}{c c c}

\AxiomC{$w = u, \Delta[w], \Delta[u], \Gamma$}
\RightLabel{$(\mathsf{sub_{=}})$}
\UnaryInfC{$w = u, \Delta[w], \Gamma$}
\DisplayProof

&

\AxiomC{$w = u, w = v, u = v, \Gamma$}
\RightLabel{$(\mathsf{eucl_{=}})$}
\UnaryInfC{$w = u, w = v, \Gamma$}
\DisplayProof

&

\AxiomC{$ \Gamma, \R_{Ag}wv, v: A$}
\RightLabel{$\agr^{*}$}
\UnaryInfC{$ \Gamma, w: [Ag] A$}
\DisplayProof

\end{tabular}
\end{center}
For $\hr$, $\agr$, $\gr$, $\ncuh$, and $\gser$ the `$\ast$' states that $v$ must be an eigenvariable.
\end{small}

\end{definition}

\noindent We note that the rules $\gconv$ and $\hconv$ express the converse relation between $\R_{\g}$ and $\R_{\h}$, and the rules $\agtd$, $\gconn$, $\hconn$, $\ncuh$ and $\{\irrtwo,\comp,$ $\comptwo \}$ correspond to conditions {\rm \textbf{(C3)}}-{\rm \textbf{(C7)}} of definition \ref{models_tkstit}, respectively. Furthermore, the notation $\Delta[u]$ in the substitution rule $(\mathsf{sub_{=}})$ is used to express a collection of relational atoms and labelled formulae where all occurrences of the label $w$ in $\Delta[w]$ have been replaced by occurrences of $u$. This notation uniformly captures all of the substitution rules given in \cite{Neg05}.
 
\begin{theorem}[Soundness]\label{Soundness_G3Tstit} Every sequent derivable in $\ldmtl$ is valid.

\end{theorem}

\begin{proof} Similar to theorem \ref{Soundness_G3LDM}.

\end{proof} 
 
\noindent Unfortunately, with respect $\ldmtl$ completeness, we cannot use the relatively simple strategy applied in proving $\ldml$ completeness. This is because the irreflexivity rule (def. \ref{TSTIT_Def}) does not readily lend itself to derivation in $\tstitl$. Here we prove $\ldmtl$ completeness relative to irreflexive $\ldmt$-frames by leveraging the methods presented in \cite{Neg09}. (NB. For this reason, we introduced $\breve{\R}_{\g}$--the complement of $\R_{\g}$--%and its rules $\comp$ and $\comptwo$ 
 directly into the language of the proof system.)% and provided the rules $\comp$ and $\comptwo$ to express their complementary relationship.)    
%The proof of completeness for $\tstitl$ is not as simple as for $\ldml$, where we only had to argue that all axioms and inference rules of $\ldm$ can be simulated in $\ldml$. 
 % 
 %The difficulty with proving completeness for $\tstitl$ in this way arises from the last inference rule given in definition \ref{TSTIT_Def}, which corresponds to the frame property of irreflexivity. This particular inference rule does not readily lend itself to derivation in $\tstitl$.   
%\textcolor{blue}{Re-check the lemma.}

\begin{lemma}\label{Provable_or_Counter} Let $\Gamma$ be a $\ldmtl$-sequent. Either, $\Gamma$ is $\ldmtl$-derivable, or it has a $\ldmt$-countermodel.

\end{lemma}

%Similar to theorem 5.4 of \cite{Neg09} (details given in appendix \ref{G3TSTIT_Completeness}).

\begin{proof} We construct the \textit{Reduction Tree} (\textbf{RT}) of a given sequent $\Gamma$, following the method of \cite{Neg09}. If \textbf{RT} is finite, all leaves are initial sequents that are conclusions of $\id$ or $\comp$. If \textbf{RT} is infinite, by K\"onig's lemma, there exists an infinite branch: %of the form
 $\Gamma_{0}$, $\Gamma_{1}$, ..., $\Gamma_{n}$,... (with $\Gamma_{0} {=}\Gamma$). Let $\mathbf{\Gamma}$ = $\bigcup \Gamma_{i}$. We define a $\ldmt$-model $M^{\ast} {=} (W, \R_{\Box}, \{R_{i} | i \in Ag\}, \R_{Ag}, \R_{\g}, \R_{\h}, V)$ as follows: Let $x \thicksim_{\mathbf{\Gamma}} y$ \textit{iff} $x {=} y \in \mathbf{\Gamma}$. (Usage of the rules ($\mathsf{ref}_{=}$) and ($\mathsf{eucl}_{=}$) in the infinite branch ensure %that
  $\thicksim_{\mathbf{\Gamma}}$ is an equivalence relation.) Define $W$ to consist of all equivalence classes $[x]$ of labels in $\mathbf{\Gamma}$ under $\thicksim_{\mathbf{\Gamma}}$. For each $\R_{\alpha} xy \in \mathbf{\Gamma}$ let $([x]_{\thicksim_{\mathbf{\Gamma}}},[y]_{\thicksim_{\mathbf{\Gamma}}})\in \R_{\alpha}$ (with $\R_{\alpha} {\in} \{\R_{\Box}, \R_{i}, \R_{Ag}, \R_{\g}, \breve{\R}_{\g}, \R_{\h}\}$), and for each labelled propositional 
 atom $x : p \in \mathbf{\Gamma}$, let $[x]_{\thicksim_{\mathbf{\Gamma}}} \not\in V(p)$. It is a routine task to show that all relations and the valuation are well-defined. Last, let the interpretation $I {:} L {\mapsto} W$ map each label $x$ to the class of labels $[x]_{\thicksim_{\mathbf{\Gamma}}}$ containing $x$, and suppose $I$ maps all other labels not in ${\mathbf{\Gamma}}$ arbitrarily. We show that: \textbf{(i)} $M^{\ast}$ is a $\ldmt$ model, and \textbf{(ii)} $M^{\ast}$ is a counter-model for $\Gamma$.

\ \ \textbf{(i)} First, we %may
 assume w.l.o.g. that $\Gamma {\neq} \emptyset$ because the empty sequent is not satisfied on any model. Thus, there must exist at least one label in $\Gamma$; i.e. %and
% hence
 $W {\neq} \emptyset$. 

\ \ We argue that $\R_{\Box}$ is an equivalence relation and omit the analogues proofs showing that $\R_{i}$ and $\R_{Ag}$ are equivalence relations. Suppose, for some $\Gamma_{n}$ in the infinite branch there occurs a label $x$ but $\R_{\Box}xx \not\in \Gamma_{n}$. By definition of \textbf{RT}, at some later stage $\Gamma_{n+k}$ the rule $\settrefl$ will be applied; hence, $\R_{\Box}xx \in \mathbf{\Gamma}$. The argument is similar for the $\setteucl$ rule. Properties {\rm \textbf{(C1)}} and  {\rm \textbf{(C2)}} follow from the rules $\bridge$ and $\ioa$, respectively. Regarding {\rm \textbf{(C3)}}, we only obtain $\R_{Ag} \subseteq \bigcap_{i \in Ag} \R_{i}$ in $M^{\ast}$ via the $\agtd$ rule. Using lemma 9 of \cite{Lor13}, we can transform $M^{\ast}$ into a model where (i) $\R_{Ag} {=} \bigcap_{i \in Ag} \R_{i}$ and where (ii) the model satisfies the same formulae.

\ \  We obtain that $\R_{\g}$ is transitive and serial due to the $\gtrans$ and $\gser$ rules. $\R_{\h}$ is the converse of $\R_{\g}$ by %the rules
 $\gconv$ and $\hconv$. The properties {\rm \textbf{(C4)}}, {\rm \textbf{(C5)}} and {\rm \textbf{(C6)}} follow from the rules $\gconn$, $\hconn$ and $\ncuh$% (resp.). %
, respectively.

\ \ {\rm \textbf{(C7)}} follows from $\irrtwo$, $\comp$, and the equality rules: these rules ensure that ($\ast$) if $[u]_{\thicksim_{\mathbf{\Gamma}}} \in \R_{\Box}([w]_{\thicksim_{\mathbf{\Gamma}}})$, then $[u]_{\thicksim_{\mathbf{\Gamma}}} \not\in \R_{\g}([w]_{\thicksim_{\mathbf{\Gamma}}})$. In what follows, we abuse notation and use $[w]$ to denote equivocally the label $w$, as well as any other label $v$ for which a chain of equalities between $w$ and $v$ occurs in the sequent. The claim ($\ast$) is obtained accordingly: if both $\R_{\Box}[w][u]$ and $\R_{\g}[w][u]$ appear together in some sequent $\Gamma_{i}$, then higher up in the infinite branch, the equality rules will introduce relational atoms of the form $\R_{\Box}w'u'$ and $\R_{\g}w'u'$. Eventually, the rule $\irrtwo$ will also be applied and, subsequently, the rule $\comp$ will ensure that the reduction tree procedure halts for the given branch. Moreover, if $\R_{\g}[w][w]$ occurs in a sequent $\Gamma_{i}$ of \textbf{RT}, then higher up in the branch the equality rules will introduce a relational atom of the form $\R_{\g}w'w'$. Eventually, $\settrefl$ will be applied which adds $\R_{\Box}w'w'$ to the branch containing $\Gamma_{i}$. Lastly, $\irrtwo$ will be applied even higher up  this branch, adding $\breve{\R}_{\g}w'w'$, which by $\comp$ will halt the \textbf{RT}-procedure for that branch. Thus we may conclude: for any infinite branch of \textbf{RT} %$\Gamma_{0}$, $\Gamma_{1}$, ..., $\Gamma_{n}$, ..., 
 $\R_{\g}ww$ will not occur for any label $w$; meaning that not only will $M^{\ast}$ satisfy {\rm \textbf{(C7)}}, its relation $\R_{\g}$ will be irreflexive. Additionally, note that $\comptwo$ will ensure that $\breve{\R}_{\g}$ is the complement of $\R_{\g}$.

%it does not occur in the infinite branch used to construct the counter-model.

% if it contained a sequent of the form $\R_{\g}xx, \Delta$.

\ \ Lastly, as long as $[x]_{\thicksim_{\mathbf{\Gamma}}} \not\in V(p)$ when $x : p \in \mathbf{\Gamma}$, %function $V$ can map 
 all other labels can be mapped by $V$ in any arbitrary manner. Thus, $V$ is a valid valuation function.

\ \ \textbf{(ii)} By construction, $M^{\ast}$ satisfies each relational atom in $\mathbf{\Gamma}$, and therefore, satisfies each relational atom in $\Gamma$. By induction on the complexity of $\phi$ it is shown that for any formula $x : \phi \in \mathbf{\Gamma}$ we have $M^{\ast}, [x]_{\thicksim_{\mathbf{\Gamma}}} \not\models \phi$ (See \cite{Neg09} for details). Hence, $\Gamma$ is falsified on $M^{\ast}$ with $I$.

\end{proof}

\begin{theorem}[Completeness]\label{Completeness_G3Tstit} Every valid sequent is derivable in $\ldmtl$.

\end{theorem}

\begin{proof} Follows from lemma \ref{Provable_or_Counter}.

\end{proof}

%{\large \textcolor{blue}{APPLICATION?}}

%\begin{corollary}
%For any formula $\phi \in \mathcal{L}_{\ldm}$ such that $\vdash_{\tstit } \phi$, it follows that $\vdash_{\ldm} \phi$, \textit{i.e.} $\tstit$ is conservative over $\ldm$.
%\end{corollary}

%\begin{proof} Let $\phi$ be an arbitrary formula in $ \mathcal{L}_{\ldm}$ such that $\vdash_{\tstit } \phi$. By theorems \ref{Soundness_Completeness_T-STIT} and \ref{Completeness_G3Tstit}, $\vdash_{\tstitl } x: \phi$ for any label $x$. Observe that none of the logical rules $\gr$, $\fr$, $\hr$, $\pr$, $\agr$, $\agdiar$ can be used in $\mathcal{D}$ because $\phi \in \mathcal{L}_{\ldm}$ and $\mathcal{D}$ has the subformula property. To complete the proof one considers the top most occurrence of a structural rule from $\tstitl$ that does not occur in $\ldml$ and shows that it can be permuted upwards and deleted. 
%\end{proof}

\section{The Logic $\xstit$}\label{XSTIT_Section}

\subsection{Axioms and Relational Semantics for $\xstit$}

A common feature of the \textit{cstit}- and \textit{dstit}-operator is that they do not internally employ temporal structures. %What is more, Lorini showed that STIT logics can be characterized by relational frames, avoiding the introduction of moment-history pairs.
 In this section, we consider the logic of $\xstit$ which contains a non-instantaneous STIT-operator explicitly affecting next states. This logic, introduced in \cite{Bro11,Bro11b}, has been motivated by the observation that affecting next states is a central aspect of agency in computer science.  %This logic uses a non-instantaneous STIT operator, which was the case for our previously considered stit-logics. 
Moreover, extensions of the logic $\xstit$ have been employed to investigate the concepts of purposeful and voluntary acts and their relation to different levels of legal culpability \cite{Bro11}. The logic %$\xstit$
 was originally proposed for a two-dimensional semantics making reference to both states and histories; the latter defined as maximally linear ordered paths on a frame. In this section, we provide a semantics for $\xstit$ %this logic
  that relies on relational frames, %which allows
  %allowing us to 
  avoiding the use of complex two-dimensional indices (the possibility of which was already noted in \cite{Bro11}). We provide a labelled calculus $\xstitl$ for this logic and prove that it is sound and complete with respect to its relational characterization. Furthermore, by showing a correspondence between the original Hilbert system $\xstit$ and the calculus $\xstitl$ we show that the language of $\xstit$ does not allow us to distinguish between the two available semantics.

\begin{definition}[The Language $\mathcal{L}_{\xstit}$]\label{language_xstit} Let $Ag {=} \{1,2,...,n\}$ be a finite set of agent labels and let $Var {=}\{p_1,p_2,p_3...\}$ be a countable set of propositional variables. %The language 
 $\mathcal{L}_{\xstit}$ is defined as follows:% defined via the following BNF grammar: 
{\small
$$\phi ::= p \ | \ \overline{p} \ | \ \phi \wedge \phi \ | \ \phi \vee \phi \ | \ \Box \phi \  | \ \Diamond \phi \ |  \ [A]^x\phi \ | \ \langle A\rangle^x \phi \ | \ [X] \phi \ | \ \langle X\rangle \phi $$}
where $p \in Var$; and $A\subseteq Ag$ (with special cases $\emptyset$ and $Ag$).
\end{definition}

%\begin{definition}[The Language $\mathcal{L}_{\xstit}$]\label{language_xstit} $Var = \{p_1,p_2,p_3...\}$ is a countable set of propositional variables. $Ag =\{a_1,a_2,a_3...\}$ is a finite set of agent labels. The language $\mathcal{L}_{\xstit}$ is defined as:% defined via the following BNF grammar: 
%{\small
%$$\phi ::= p \ | \ \overline{p} \ | \ \phi \wedge \phi \ | \ \phi \vee \phi \ | \ \Box \phi \  | \ \Diamond \phi \ |  \ [A]^x\phi \ | \ \langle A\rangle^x \phi \ | \ [X] \phi \ | \ \langle X\rangle \phi $$}
%where $p \in Var$, $A\subseteq Ag$ (with special cases $\emptyset$ and $Ag$) and $\top = p\lor\overline{p}$ and $\bot = p\land \overline{p}$.
%\end{definition}
%Furthermore, we define the operators $\Diamond$ and $\[A]^x$, respectively, as follows: $\Diamond \phi =_{def} \lnot\Box\lnot \phi$ and $\langle A\rangle^x \phi =_{def} \lnot [A]^x \lnot \phi$. 
%*Note that the $X$ operator is self-dual.
\noindent The language uses the %familiar
 settledness operator $\Box$, a group-stit operator $[A]^{x}$, and the operator $[X]$ referring to the next state. Formulae of the form $[A]^{x}\phi$ must be read as `group $A$ effectively sees to it that in the next state $\phi$ holds'. % and for a natural language interpretation of the axioms we refer to that paper. 

As mentioned previously, we provide a semantics for the logic $\xstit$ based on relational frames. The conditions on these frames are obtained through a simple transformation of the two-dimensional frame properties presented in \cite{Bro11}.

\begin{definition}[Relational $\xstit$ Frames and Models]\label{models_xstit} An $\xstit$-frame is defined to be a tuple $F = ( W, \R_{\Box}, \R_X, \{\R_A | A\subseteq Ag\})$ such that $W \neq \emptyset$ and: %is a non-empty set of worlds and: 

\begin{itemize}
 
\item[{\rm \textbf{(D1)}}] $\R_{\Box}\subseteq W{\times}W$ is an equivalence relation;

\item[{\rm \textbf{(D2)}}] $\R_X\subseteq W{\times}W$ is serial and deterministic;

\item[{\rm \textbf{(D3)}}] $\R_A \subseteq W{\times}W$ such that,

\begin{itemize}

\item[(i)] $\R_{\emptyset} = \R_{\Box} \circ \R_X$;

\item[(ii)] $\R_{Ag} = \R_X \circ \R_{\Box}$;

\item[(iii)] $\R_A\subseteq \R_B$ for $\emptyset \subseteq B\subseteq A\subseteq Ag$;

\item[(iv)] For any $B,A\subseteq Ag$ (s.t. $B\cap A=\emptyset$) and $\forall w_1,w_2,w_3,w_5,w_6\in W$ we have: $(\R_{\Box} w_1w_2\land \R_{\Box}w_1w_3) \rightarrow \exists w_4 (\R_{\Box}w_1w_4 \land (\R_A w_4w_5 \rightarrow \R_Aw_2w_5)\land (\R_B w_4w_6 \rightarrow \R_Bw_3w_6))$ %(IOA$^x$) 

\end{itemize}

\end{itemize}

\noindent %Let $\mathcal{C}^{\xstit}$ be the class of relational $\xstit$-frames. 
A relational $\xstit$-model is a tuple $M = ( F,V )$ where $F$ is an $\xstit$-frame and $V$ a valuation function mapping propositional variables $p_i\in Var$ to subsets of $W$; i.e. $V: Var\mapsto \mathcal{P}(W)$.
\end{definition}

\noindent 
Condition (D3)-(iv) expresses the \textit{independence of agents} principle for $\xstit$. From condition (D3)-(ii) we obtain that $\mathcal{R}_{Ag} \subseteq \mathcal{R}_X \circ \mathcal{R}_{\Box}$, which ensures the principle of \textit{no choice between undivided histories} (\textit{cf}. definition \ref{models_tkstit}, C6). Furthermore, we stress that, following \cite{Bro11}, the relation $\R_X$ is not explicitly defined as a \textit{strict} next-relation; that is, the frame construction allows for reflexive worlds. For a discussion of all the frame properties we refer the reader to \cite{Bro11}. %This means that
% I DELETED THIS SENTENCE SINCE IT HAS NO FUNCTION. And I am not really sure whether it is correct. Actually I believe that it is not... Furthermore, at such a reflexive world $w$ only one `choice' is available, namely, the one leading to $w$ itself (\textit{i.e.} $\R_{\emptyset}(w) = \R_{Ag}(w)$). IS THIS ACTUALLY CORRECT? %, no real choices are available \textit{i.e.} $\R_{\emptyset}(w) = \R_{Ag}(w)$.

%I'm not sure. What seems weird is that if you have reflexive points, then for a given moment (i.e., white box equivalence class), the choices available to a set of agents seems to include both points in the current moment and next moment, which sort of defeats the purpose of having a logic based on realizing states in the next moment.

\begin{definition}[Semantics of $\mathcal{L}_\xstit$] 
To define the \emph{satisfaction} of a formula $\phi\in\mathcal{L}_{\xstit}$ on $M$ at $w$, we make use of clauses (1)-(6) from definition \ref{Semantics_ldm_tstit}, taking $M$ to be an $\xstit$-model (but omitting explicit mention of $M$ in the clauses), along with the following clauses (global truth and validity are defined as usual):
\begin{center}
\begin{small}
\begin{multicols}{2}
\begin{itemize}
\item[7.] $w \models [A]^x \phi$ iff  $\forall u \in \R_{A}(w)$, $u \models \phi$;

\item[8.] $w \models \langle A\rangle^x \phi$ iff $\exists u \in \R_{A}(w)$, $u \models \phi$;
\end{itemize}
\columnbreak
\begin{itemize}

\item[9.] $w \models [X] \phi$ iff $\forall u \in \R_{X}(w)$, $u \models \phi$;

\item[10.] $w \models \langle X\rangle \phi$ iff $\exists u \in \R_{X}(w)$, $u \models \phi$.

\end{itemize}
\end{multicols}
\end{small}
\end{center}
\end{definition}

\begin{definition}[The Logic $\xstit$ \cite{Bro11}]\label{axioms_xstit}
The Hilbert system for $\xstit$ consists of %all propositional tautologies, modus ponens and 
 the axioms and rules below, where % the following axiom schemes and rules (we use $\vdash_{\xstit}$ to indicate a Hilbert style derivation). 
$\phi, \psi\in\mathcal{L}_{\xstit}$, $A\subseteq Ag$ and $\alpha \in \{\Box,[A]^x,[X]\}$:

\begin{small}
\begin{center}
\begin{tabular}{c @{\hskip 2em} c @{\hskip 2em} c} 

$\phi\rightarrow (\psi \rightarrow \phi)$

&

$(\overline{\psi} \rightarrow \overline{\phi}) \rightarrow (\phi \rightarrow \psi)$

&

$(\phi \rightarrow (\psi \rightarrow \chi)) \rightarrow ((\phi \rightarrow \psi) \rightarrow (\phi \rightarrow \chi))$

\end{tabular}

\ \\
%\end{center}
%\begin{center}

\begin{tabular}{c @{\hskip 1.5em} c @{\hskip 1.5em} c @{\hskip 1.5em} c @{\hskip 1.5em} c @{\hskip 1.5em} c }
$\alpha (\phi\rightarrow \psi)\rightarrow (\alpha \phi\rightarrow \alpha \psi)$ 
&

$\Box \phi \rightarrow \phi$

&

$\Diamond \phi \rightarrow \Box \Diamond \phi$

%$\Box (\phi\rightarrow \psi)\rightarrow (\Box\phi\rightarrow \Box \psi)$ 

%& 

%	$[A]^x (\phi \rightarrow \psi) \rightarrow ([A]^x \phi \rightarrow [A]^x \psi)$

&

$[A]^x\phi \rightarrow \langle A\rangle^x \phi$
&
$\diamx \phi \rightarrow [X] \phi$

\end{tabular}

\ \\
%\end{center}
%\begin{center}
\begin{tabular}{ c @{\hskip 1.5em} c @{\hskip 1.5em} c@{\hskip 1.5em} c @{\hskip 1.5em} c @{\hskip 1.5em} c}

%$[X] (\phi\rightarrow\psi)\rightarrow ([X]\phi\rightarrow [X]\psi)$

$\Box [X]\phi \leftrightarrow [\emptyset]^x \phi$

& 

$[Ag]^x \phi \leftrightarrow [X]\Box\phi$

& 

$[A]^x \phi \rightarrow [B]^x \phi^{(\dag)}$
&

$\Box \phi \vee \Diamond \overline{\phi}$ 

&

$[A]^x \phi \vee \langle A\rangle^x \overline{\phi}$

\end{tabular}

\ \\

%\end{center}
%\begin{center}
\begin{tabular}{c @{\hskip 1.5em} c @{\hskip 1.5em} c @{\hskip 1.5em} c @{\hskip 1.5em} c }

 $\Diamond [A]^x \phi \land \Diamond [B]^x \psi \rightarrow \Diamond ([A]^x \phi \land [B]^x \psi)^{(\dag\dag)}$

&

$[X] \phi \vee \langle X\rangle \overline{\phi}$
& 
\AxiomC{$\phi$}
\AxiomC{$\phi \rightarrow \psi$}
\BinaryInfC{$\psi$}
\DisplayProof

&
\AxiomC{$\phi$}
\UnaryInfC{$\alpha \phi$}
\DisplayProof

%&
%\\

%$\Box X \phi \rightarrow \agbox{}^{x} \phi$

%&

%$\agbox{}^{x} \phi \rightarrow \agdia{}^{x} \phi$

%&

%&

%$\agbox{}^{x} \phi \rightarrow X \Box \phi$

%&

%\AxiomC{$\phi$}
%\UnaryInfC{$\Box \phi$}
%\DisplayProof

%&%
%\AxiomC{$\phi$}
%\UnaryInfC{$[A]^x \phi$}
%\DisplayProof
%&
%\AxiomC{$\phi$}
%\UnaryInfC{$[X] \phi$}
%\DisplayProof

\end{tabular}
\end{center}

\noindent where $(\dag) A\subseteq B\subseteq Ag$; and $(\dag\dag) A \cap B =\emptyset$. 
\end{small}
\\
\\
\noindent A derivation of $\phi$ in $\xstit$ from $\Theta$ is written as $\Theta \vdash_{\xstit} \phi$. When $\Theta$ is the empty set, we refer to $\phi$ as a \emph{theorem} and write $\vdash_{\xstit} \phi$.
\end{definition}
%Originally, frames were introduced for $\xstit$ based on a semantics of dynamic states, employing complex indexes built from states together with histories . In the light of Lorini.... 
We refer to $\Diamond [A]^x \phi \land \Diamond [B]^x \psi \rightarrow \Diamond ([A]^x \phi \land [B]^x \psi)$ as the IOA$^x$-axiom. In contrast with the standard IOA-axiom, observe that IOA$^x$-axiom refers to the independence of isolated \textit{groups} of agents with respect to \textit{successor} states. For a natural language interpretation of the other axioms of $\xstit$ we refer to \cite{Bro11}.

Instead of proving completeness for the intended sequent calculus directly, we prove it first for the Hilbert calculus. This enables us to eventually conclude the equivalence of these two calculi with respect to the logic $\xstit$. %the language of $\mathcal{L}_{\xstit}$.

\begin{theorem}[Completeness of $\xstit$]\label{xstitcomp} 
For all $\phi\in\mathcal{L}_{\xstit}$, if $\models \phi$, then $\vdash_{\xstit} \phi$. 
 \end{theorem}

\begin{proof}
As observed in \cite{Bro11}, all axioms of $\xstit$ are Sahlqvist formulae%, since IoA is what is called a Simple Sahlqvist Formula)
. Furthermore, the first-order correspondents of the $\xstit$ axioms %is obtained via the algorithm in Theorem 3.49 of \cite{BlaRijVen01} and all 
 taken together define the class of frames from definition \ref{models_xstit}. %Since every axiom of $\xstit$ is canonical relative to the first-order property which it defines, 
Applying Theorem 4.42 of \cite{BlaRijVen01}, we obtain that the logic $\xstit$ is %strongly
complete relative to this class of frames.

%As observed in \cite{Bro11}, all axioms of $\xstit$ are Sahlqvist formulae%, since IoA is what is called a Simple Sahlqvist Formula)
%. The first-order correspondents of all axioms of $\xstit$ %is obtained via the algorithm in Theorem 3.49 of \cite{BlaRijVen01} and all 
% taken together define the class of frames as in definition \ref{models_xstit}. Since every axiom of $\xstit$ is canonical relative to the first-order property which it defines, we know that the logic $\xstit$ is strongly complete relative to \textcolor{red}{this class of frames} (Theorem 4.42 of \cite{BlaRijVen01}).

 %is shown for all normal modal calculi extended with generalized geometric rules, which agrees with the logics axiomatized by corresponding Sahlqvist formulae. Here, .}    
\end{proof}

\subsection{A Cut-free Labelled Calculus for $\xstit$}

We provide a labelled calculus $\xstitl$ that is sound %with respect to the frame-class $\mathcal{C}_{\xstit}$
 and complete relative to the relational frames of definition \ref{models_xstit}. In order to convert the $\xstit$ axiomatization into rules for the intended calculus, we first observe that every axiom of $\xstit$ is a \textit{geometric formula} with the exception of the IOA$^x$ axiom. % $\Diamond [A]^x \phi \land \Diamond [B]^x \psi \rightarrow \Diamond ([A]^x \phi \land [B]^x \psi)$.
  For the geometric formulae %of $\xstit$
   we can find corresponding geometric rules, following \cite{Neg05}. The first-order frame condition $\mathrm{(D3)}(iv)$ for IOA$^x$ (def. \ref{models_xstit}) is not a geometric formula; however, we observe that its components $\R_{A}w_{4}w_{5} {\rightarrow} \R_{A}w_{2}w_{5}$ and $\R_{B}w_{4}w_{6} {\rightarrow} \R_{B}w_{3}w_{6}$ in fact are. %geometric formulae. %(in the class $GA_{0}$). 
 The IOA$^x$-condition is, thus, a \textit{generalized geometric axiom} of type $GA_1$ and we may therefore find an equivalent system of rules, following \cite{Neg16}. %these results apply to $\ldml$ and $\xstitl$ as well. 
  
  %\textcolor{red}{In addition, we extend the closure condition of section \ref{cutfreeldm} to apply to generalized geometric rules.DELETE???}
% } %\textcolor{purple}{Recall Theorem 1 here!}
%begin{small}
%begin{center}
%begin{tabular}{c c c}
%\AxiomC{$R_{\Box}w_{1}w_{2}, R_{\Box}w_{1}w_{3}, R_{\Box}w_{1}w_{4}, \Gamma$}
%\RightLabel{($\mathsf{IOA{-}E})^{*}$}
%\UnaryInfC{$R_{\Box}w_{1}w_{2}, R_{\Box}w_{1}w_{3},\Gamma$}
%\DisplayProof
%&
%\AxiomC{$R_{A}w_{4}w_{5}, R_{A}w_{2}w_{5}, \Gamma$}
%\RightLabel{($\mathsf{IOA{-}U_{1}})$}
%\UnaryInfC{$R_{A}w_{4}w_{5}, \Gamma$}
%\DisplayProof
%&
%\AxiomC{$R_{B}w_{4}w_{6}, R_{B}w_{3}w_{5}, \Gamma$}
%\RightLabel{($\mathsf{IOA{-}U_{2}})$}
%\UnaryInfC{$R_{B}w_{4}w_{6}, \Gamma$}
%\DisplayProof
%\end{tabular}
%\end{center}
%\end{small}
%\noindent
%(*): $w_{4}$ is an eigenvariable in the ($\mathsf{IOA-E}$) rule.
%\\

%Fig. \ref{ioasystem}).

We refer to the following system of rules $\langle (\mathsf{IOA{-}E}), \{(\mathsf{IOA{-}U_1}),(\mathsf{IOA{-}U_{2}})\}\rangle$ as the `independence of agents' rule $\ioax$. We may use the rule ($\mathsf{IOA{-}E}$) wherever throughout the course of a derivation, but if we use either ($\mathsf{IOA{-}U_1}$) or ($\mathsf{IOA{-}U_2}$), then we must (i) use the other ($\mathsf{IOA{-}U}_{i}$) rule (for $i \in \{1,2\}$) in a separate branch of the derivation and (ii) use the ($\mathsf{IOA{-}E}$) rule below both instances of ($\mathsf{IOA{-}U}_{i}$); \textit{i.e.} the derivation is of the form represented below: %in fig. \ref{ioasystem}. %; %here must be an instance of the IOA-E rule; 
%\textit{i.e.}, our derivation will contain an instance of the system $\ioax$ in following form:

%\begin{center}
%\[
%(\mathsf{IOA}) =
%  \begin{cases}
%              & (\mathsf{IOA-U_{1}}), (\mathsf{IOA-U_{2}})  \\
%              & \\%
%             & (\mathsf{IOA-E})
%\end{cases}
%\]
%\end{center}

%\noindent
%where:\\

\begin{small}
%\vspace{-0.3cm}
%\begin{figure}

\begin{center}

\AxiomC{$R_{A}w_{4}w_{5}, R_{A}w_{2}w_{5}, \Gamma$}
\RightLabel{($\mathsf{IOA-U_{1}})$}
\UnaryInfC{$R_{A}w_{4}w_{5}, \Gamma$}
\UnaryInfC{$\vdots$}

\AxiomC{$R_{B}w_{4}w_{6}, R_{B}w_{3}w_{6}, \Gamma'$}
\RightLabel{($\mathsf{IOA-U_{2}})$}
\UnaryInfC{$R_{B}w_{4}w_{6}, \Gamma'$}
\UnaryInfC{$\vdots$}

\BinaryInfC{$R_{\Box}w_{1}w_{2}, R_{\Box}w_{1}w_{3}, R_{\Box}w_{1}w_{4}, \Gamma''$}
\RightLabel{($\mathsf{IOA-E})^{*}$}
\UnaryInfC{$R_{\Box}w_{1}w_{2}, R_{\Box}w_{1}w_{3},\Gamma''$}
\DisplayProof

\end{center}
\ \\
\noindent
where (*) $w_{4}$ is an eigenvariable in the ($\mathsf{IOA-E}$) rule.
%\caption{The system of rules $\ioax$.}\label{ioasystem}
%\end{figure}

\end{small}

%The labelled calculus for the logic $\xstit$ is presented in definition \ref{sequentxstitl}. 

%\vspace{-0.7cm}
\begin{definition}[The Calculus $\xstitl$]\label{sequentxstitl} The labeled calculus $\xstitl$ consists of the rules $\id$, $\conr$, $\disr$, $\reflis$, $\euclis$, $\subis$, $\settr$, $\settdiar$, $\settrefl$, and $\setteucl$ from definitions \ref{sequentldml} and \ref{sequentldmtl} extended with the $\ioax$-rule %= \langle (\mathsf{IOA-E}), \{(\mathsf{IOA-U_1}),(\mathsf{IOA-U_2})\}\rangle$ rule 
and the following: % set: % of rules:  

\begin{small}
%\begin{center}
%\begin{tabular}{c @{\hskip 1 em} c @{\hskip 1 em} c @{\hskip 1 em} c}

%\AxiomC{ }
%\RightLabel{$\id$}
%\UnaryInfC{$\Gamma, w:p, w:\overline{p} $}
%\DisplayProof

%&

%\AxiomC{$\Gamma, w: \phi$}
%\AxiomC{$\Gamma, w: \psi$}
%\RightLabel{$\conr$}
%\BinaryInfC{$\Gamma, w: \phi \wedge \psi$}
%\DisplayProof

%&

%\AxiomC{$\Gamma, w: \phi, w : \psi$}
%\RightLabel{$\disr$}
%\UnaryInfC{$\Gamma, w: \phi \vee \psi$}
%\DisplayProof

%& 

%\AxiomC{$\Gamma, \R_{\Box}wu, u: \phi$}
%\RightLabel{$\settr^{*}$}
%\UnaryInfC{$\Gamma, w: \Box \phi$}
%\DisplayProof
%\end{tabular}
%\end{center}

\begin{center}
\begin{tabular}{c @{\hskip 1 em} c @{\hskip 1 em} c @{\hskip 1 em} c}

%\AxiomC{$\Gamma, \R_{\Box}wu, w: \Diamond \phi, u: \phi$}
%\RightLabel{$\settdiar$}
%\UnaryInfC{$\Gamma, \R_{\Box}wu, w: \Diamond \phi$}
%\DisplayProof

%&

\AxiomC{$\Gamma, \R_A wv, v: \phi$}
\RightLabel{$\stitboxx^{*}$}
\UnaryInfC{$\Gamma, w: [A]^x \phi$}
\DisplayProof

&

\AxiomC{$\Gamma, \R_Awu, w: \lb A \rb^x \phi, u: \phi$}
\RightLabel{$\stitdiamx$}
\UnaryInfC{$\Gamma, \R_A wu, w: \lb A \rb^x \phi$}
\DisplayProof

\end{tabular}
\end{center}

\begin{center}
\begin{tabular}{c c}

\AxiomC{$ \Gamma, \R_X wv, w: \diamx \phi, v: \phi$}
\RightLabel{$\diamnext$}
\UnaryInfC{$\Gamma, \R_X wv  w: \langle X\rangle \phi$}
\DisplayProof

&

\AxiomC{$\R_{\Box}wv, \R_X vu, \R_{\emptyset} wu,  \Gamma$}
\RightLabel{$\effempty$}
\UnaryInfC{$\R_{\Box}wv, \R_X vu, \Gamma$}
\DisplayProof

\end{tabular}
\end{center}

\begin{center}
\begin{tabular}{ c @{\hskip 1em} c @{\hskip 1em} c }

\AxiomC{$\R_A wv, \R_B wv, \Gamma$}
\RightLabel{$\cmon^{\dag}$}%, $B\subseteq A$}
\UnaryInfC{$\R_A wv, \Gamma$}
\DisplayProof
&
\AxiomC{$\R_X wv,  \Gamma$}
\RightLabel{$\serx^{\ast}$}
\UnaryInfC{$\Gamma$}
\DisplayProof
& 
\AxiomC{$v = u, \R_X wv, \R_X wu,  \Gamma$}
\RightLabel{$\detx$}
\UnaryInfC{$\R_X wv,\R_X wu,  \Gamma$}
\DisplayProof

\end{tabular}
\end{center}
%\begin{center}
%\begin{tabular}{c @{\hskip 1.5em} c @{\hskip 1.5em} c }

%\AxiomC{$w = w,  \Gamma$}
%\RightLabel{$\reflis$}
%\UnaryInfC{$ \Gamma$}
%\DisplayProof

%&

%\AxiomC{$w = v, w = u, v = u,  \Gamma$}
%\RightLabel{$\euclis$}
%\UnaryInfC{$w = v, w = u,  \Gamma$}
%\DisplayProof

%&
%\AxiomC{$w = v, \Q[w], \Delta[w], \Q[v], \Delta[v],  \Gamma$}
%\RightLabel{$\subis$}
%\UnaryInfC{$w = v, \Q[w], \Delta[w],  \Gamma$}
%\DisplayProof

%\end{tabular}

%\end{center}

%\textcolor{purple}{The restriction on variable introduction should be something like: $wR_{\Box}v$ iff not $v$ occurs in a next or previous moment.Very Important! Otherwise we get contradictions!}
\begin{center}

\begin{tabular}{c @{\hskip 1em} c @{\hskip 1em} c }

\AxiomC{$ \Gamma, \R_X wv, v: \phi$}
\RightLabel{$\boxnext^{\ast}$}
\UnaryInfC{$\Gamma, w: [X] \phi$}
\DisplayProof

&

\AxiomC{$\R_{\Box}wv, \R_X vu, \R_{\emptyset} wu,  \Gamma$}
\RightLabel{$\emptyeff ^{\ast}$}%, $v\not\in Con$ or $v\in Now_w$}
\UnaryInfC{$\R_{\emptyset} wu,  \Gamma$}
\DisplayProof

\end{tabular}
\end{center}

\begin{center}

\begin{tabular}{c @{\hskip 2.5em} c }

\AxiomC{$\R_{Ag}wu, \R_X wv, \R_{\Box}vu,  \Gamma$}
\RightLabel{$\effag$}
\UnaryInfC{$\R_X wv, \R_{\Box}vu,  \Gamma$}
\DisplayProof

&

\AxiomC{$\R_{Ag}wu, \R_X wv, \R_{\Box}vu,  \Gamma$}
\RightLabel{$\ageff^{\ast}$}%, $v\not\in Con$ or $v\in Now_u$}
\UnaryInfC{$\R_{Ag}wu, \Gamma$}
\DisplayProof

%\begin{center}
%\AxiomC{$\R_{\Box}wv, \R_X vu, \R_A wu,  \Gamma$}
%\RightLabel{$v \not \in Con.$}
%\UnaryInfC{$\R_A wu,  \Gamma$}
%\DisplayProof

%&

%\AxiomC{$\R_A wu, \R_X wv, \R_{\Box}vu,  \Gamma$}
%\RightLabel{$\ncuhx$, where $A\subseteq Ag$}
%\UnaryInfC{$\R_X wv, \R_{\Box}vu ,  \Gamma$}
%\DisplayProof

%\AxiomC{$R^{x}_{\alpha}(m,h)(m',h'), R^{x}(m,h)(m',h), R^{m'}_{\Box}hh',  \Gamma$}
%\UnaryInfC{$R^{x}(m,h)(m',h), R^{m'}_{\Box}hh',  \Gamma$}
%\DisplayProof
%\end{center}

\end{tabular}
\end{center}
\noindent where $(\ast)$ $v$ is an eigenvariable; and $(\dag)$ $B\subseteq A\subseteq Ag$.
\end{small}
\end{definition}

%\noindent Observe that the rules $\{\emptyeff,\effempty\}, \{\ageff,\effag\}, \cmon$ and $\ioax$ of the %labelled
% calculus $\xstitl$ capture the frame conditions expressed in $\mathrm{(D3)}(i){-}(iv)$, respectively.
 
\noindent Observe that the rules $\{\emptyeff,\effempty\}, \{\ageff,\effag\}, \cmon$ and $\ioax$ of the labelled calculus $\xstitl$ capture the frame conditions $\mathrm{(D3)}(i){-}(iv)$ of definition \ref{models_xstit}, respectively.\footnote{In \cite{Neg16} it is shown that every generalized geometric formula can be captured through (a system of) rules, allowing %which allows
 for the construction of \textit{analytic} calculi for the minimal modal logic $\mathsf{K}$ extended with any axioms from the Sahlqvist class. Since all axioms of $\ldm$ and $\xstit$ are Sahlqvist formulae, the results also apply to these logics.}

%*Note that $\Q[m]$ and $\Delta[m]$ represent relational atoms and labelled formula containing the label $w$, respectively. We can decompose this rule into separate rules for each type of relational atom, etc. but it is easier to just have one rule that uniformly covers all cases (this rule is also in the $\ldmtl$ calculus, so I will introduce the notation for substituting there).

\begin{theorem}[Soundness]\label{xstitsound} Every sequent derivable in $\xstitl$ is valid.
\end{theorem}
\begin{proof} Similar to theorem \ref{Soundness_G3LDM}. %The only distinct cases are the geometric rules $\detx, \emptyeff, \effempty, \ageff, \effag$ which are all explicitly captured through the frame properties of $\mathcal{C}_{\xstit}$ (cf. \ref{negri_kripke}). The only peculiar case is the one related to the independence of agents rule. 
 Since all rules of $\xstitl$ are generalized geometric rules, we can apply the general soundness results of Theorem 6.3 of \cite{Neg16}.% to obtain soundness for $\xstitl$.   
\end{proof}

\noindent %The following simple frame ensures that the class $\mathcal{C}^{\xstit}$ is non-empty: $\mathcal{F} = \langle W, R_X, R_{\Box}, \{R_{A} | A\subseteq Ag \} \rangle$ for $Ag = \{a,b\}$. Let $ W= \{w_1,w_2\},  R_X=\{(w_1,w_2),(w_2,w_2)\}, R_{\Box} =\{(w_1,w_1),(w_2,w_2)\}$. Let $R_{\{a\}}=R_{\{b\}}=R_{Ag} = R_{\emptyset} =\{(w_1,w_2),(w_2,w_2)\}$. The frame $F$ satisfies definition \ref{models_xstit}. 
 In order to prove completeness of $\xstitl$ relative to the logic $\xstit$, we employ the same strategy as for $\ldml$, by first proving that every formula derivable in $\xstit$ is derivable in $\xstitl$. 

\begin{lemma}\label{xstitsequent}
For all $\phi \in \mathcal{L}_{\xstit}$, if $\vdash_{\xstit} \phi$, then $\vdash_{\xstitl} x : \phi$. 
\end{lemma}

\begin{proof} The derivation of each axiom and inference rule %(with the exception of the IOA axiom)
 is straightforward %exercise
  (See \cite{Neg05}). The $\xstitl$-derivation of the IOA$^x$-axiom can be obtained by applying the rule system $\ioax$ (see appendix \ref{XSTIT_IOA}).
\end{proof}

\begin{corollary}[Completeness]\label{xstitlcomp} For all $\phi\in\mathcal{L}_{\xstit}$, if $\models \phi$, then $\vdash_{\xstitl} x : \phi$ %Every valid formula is derivable in $\xstitl$. %For all $\phi \in \mathcal{L}_{\xstit}$, if $\models \phi$, then $\vdash_{\xstitl} x : \phi$. %Let $\mathcal{C}^{dyn}$ stand for the class of dynamic frames defined in \cite{Bro}. The logic $\xstit$ does not distinguish between static and dynamic frames: 
%$ \mathcal{C}^{stat}{\models} \phi \iff  \mathcal{C}^{dyn}{\models} \phi$.
\end{corollary}
\begin{proof}
Follows from theorem \ref{xstitcomp} and lemma \ref{xstitsequent}.
\end{proof}
%  
%\begin{corollary}[Completeness of $\xstitl$% and equivalence to $\xstit$
%]\label{xstitlcomp} For all $\phi \in \mathcal{L}_{\xstit}$, if $\models \phi$, then $\vdash_{\xstitl} x : \phi$. %Let $\mathcal{C}^{dyn}$ stand for the class of dynamic frames defined in \cite{Bro}. The logic $\xstit$ does not distinguish between static and dynamic frames: 
%%$ \mathcal{C}^{stat}{\models} \phi \iff  \mathcal{C}^{dyn}{\models} \phi$.
%\end{corollary}
%\begin{proof}
%Follows from theorem \ref{xstitcomp} and lemma \ref{xstitsequent}.
%\end{proof}
%\noindent Moreover, from theorems \ref{xstitsound}, \ref{xstitcomp} and lemma \ref{xstitsequent} we conclude that, \textit{for all} $\phi \in \mathcal{L}_{\xstit}$, $\vdash_{\xstit} \phi$ \textit{iff} $\models \phi$ \textit{iff} $\vdash_{\xstitl} x : \phi$. 
As another consequence, we obtain that the logic $\xstit$ can be characterized without using % distinguish between relational frames and 
 %more complex 
 two-dimensional frames employing histories, as applied in \cite{Bro11}. %\\
% \textcolor{red}{check Round up the whole thingy here\\

%\section{The Logic $\atstit$}\label{ATSTIT_Section}

%   \input{Section-4.tex}
   
\section{Conclusion and Future Work}\label{Conclusion}

   In this paper, we laid the proof-theoretic foundations %preluded the quest for % addressed the problem of
 for implementable logics of agency %for automated reasoning systems
by providing calculi for one of its central formalisms: STIT logic. In particular, we developed cut-free labelled sequent calculi for three STIT logics: $\ldm$, $\ldmt$ and $\xstit$. Furthermore, by providing the cut-free calculus $\tstitl$ for temporal STIT logic we answered the open question from \cite{Wan06}. All labelled calculi presented in this work, are sound and cut-free complete relative to their classes of \textit{temporal relational} frames. As a corollary to the latter, we extended prior results from \cite{BalHerTro08,HerSch08,Lor13} and provided a characterization of $\xstit$ through relational frames. 

%showed that the logic $\xstit$ is also characterizable through relational frames. %semantics using histories.

%In this paper we provided cut-free labelled sequent calculi for three STIT logics: $\ldm$, $\ldmt$ and $\xstit$. In particular, by providing the analytic calculus $\tstitl$ for the temporal STIT logic we answered the open question from \cite{Wan06}.  All labelled calculi presented here are sound and cut-free complete relative to their classes of temporal relational frames. As a corollary to the latter, extending prior results from \cite{Lor13,HerSch08}, we obtained that the logic $\xstit$ is characterizable through relational frames. %semantics using histories.

We see two possible future extensions of the calculi provided in this paper: First, we aim to use these calculi to solve the decidability problems for $\ldmt$ and $\xstit$, which are currently open questions. Our approach will be proof-theoretic in nature and will consist of showing decidability via proof-search. To realize our goal, we plan on harnessing refinement (i.e. internalization) procedures, such as those in \cite{CiaLyoRamTiu19}, to obtain variants of our labelled calculi that are more suitable for proof-search. %new variants of our labelled calculi that are better suited for proof-search.
%we see potential in harnessing translation procedures, such as those in \cite{CiaLyoRam18}, to refine the proposed labelled calculi and make them more appropriate for proof search. It is known that the results in \cite{CiaLyoRam18} can be extended to refine labelled calculi for larger classes of multi-modal logics. 
%In fact, the single-agent version of $\ldm$ already fits within a class of logics for which refinement of labelled calculi is known, suggesting the potential %gives rise to the possibility that we may be able 
% of extending these refinement methods to the multi-agent setting for $\ldm$.
Second, %Additionally, 
 we aim to extend the current calculi to incorporate formal concepts that enable reasoning about normative choice-making, for example, those found in utilitarian deontic STIT \cite{Hor01,Mur04} and legal theory \cite{LorSar15}.
\\
\\
\textbf{Acknowledgments.}  %Work supported by WWTF project MA16-028 and the FWF projects I2982 and W1255-N23. 
Work funded by the projects WWTF MA16-028, FWF I2982 and FWF W1255-N23.
 The authors would %also
  like to thank their supervisor Agata Ciabattoni for her %support and
 helpful comments.
% FWF project W1255-N23 (LogiCS: Check!)

% extensions of STIT logic, \textcolor{purple}{such as utilitarian deontic STIT \cite{Hor01,Mur04} and legal concepts \cite{LorSar15}.} %Last, we aim to further investigate the relationship between the standard achievement stit operator and the alternative achievement stit operator proposed in this paper.

%This result  by showing that $\xstit$ can be characterized through relational frames, 

 \bibliographystyle{splncs04}

%%%%%%%%%%%%%%%%%%%%%%%%%%%%%%%%%%%%%%%%%%%%
%%%%%%%%%%%%%%%%%APPENDIX%%%%%%%%%%%%%%%%%%%
%%%%%%%%%%%%%%%%%%%%%%%%%%%%%%%%%%%%%%%%%%%%

\appendix

   \section{Completeness of $\ldm$}\label{Completeness_Proof_ldm}
%In this appendix, 
We give the definitions and lemmas sufficient to prove the completeness of $\ldm$ relative to $\tstit$ frames \cite{Lor13,BerLyo19}. We make use of the canonical model of $\ldm$ (obtained by standard means \cite{BlaRijVen01,BalHerTro08}) to construct a $\tstit$ model. A truth-lemma is then given relative to this model, from which, completeness follows as a corollary.

\begin{definition}[$\ldmn$-CS, $\ldmn$-MCS] A set $\Theta \subset \mathcal{L}_{\ldmn}$ is a \emph{$\ldmn$ consistent set ($\ldmn$-CS)} iff $\Theta \not\vdash_{\ldmn} \bot$. We call a set $\Theta \subset \mathcal{L}_{\ldmn}$ a \emph{$\ldmn$ maximally consistent set ($\ldmn$-MCS)} iff $\Theta$ is a $\ldmn$-CS and for any set $\Theta'$ such that $\Theta \subset \Theta'$, $\Theta' \vdash_{\ldmn} \bot$.
\end{definition}

\begin{lemma}[Lindenbaum's Lemma \cite{BlaRijVen01}]\label{Lindenbaum}
Every $\ldmn$-CS can be extended to a $\ldmn$-MCS.
\end{lemma}

\begin{definition}[Present and Future Pre-Canonical $\tstit$ Model] The \emph{present pre-canonical $\tstit$ model} is the tuple $M^{\pres} = (W^{\pres}, \R^{\pres}_{\Box},$ $\{\R^{\pres}_{i}|i \in Ag\}, V^{\pres})$ defined below left, and the \emph{future pre-canonical $\tstit$ model} is the tuple $M^{\fut} = (W^{\fut}, \R^{\fut}_{\Box}, \{\R^{\fut}_{i}|i \in Ag\}, V^{\fut})$ defined below right:

\begin{center}
\begin{multicols}{2}

\begin{center}

\begin{itemize}

\item $W^{\pres}$ is the set of all $\ldm$-MCSs;

\item $\R^{\pres}_{\Box}wu$ iff for all $\Box \phi \in w$, $\phi \in u$;

\item $\R^{\pres}_{i}wu$ iff for all $[i] \phi \in w$, $\phi \in u$;

\item $V^{\pres}(p) = \{w \in W| p \in w\}$.

\end{itemize}
\end{center}

\begin{center}

\begin{itemize}

\item $W^{\fut} = W^{\pres}$;

\item $\R^{\fut}_{\Box}(w) = \bigcap_{i \in Ag} \R^{\pres}_{i}(w)$;

\item $\R^{\fut}_{i}(w) = \bigcap_{i \in Ag} \R^{\pres}_{i}(w)$;

\item $V^{\fut}(p) = V^{\pres}(p)$.

\end{itemize}
\end{center}

\end{multicols}
\end{center}

\end{definition}

\begin{definition}[Canonical Temporal Kripke STIT Model] We define the \emph{canonical temporal Kripke STIT model} to be the tuple $M^{\ldm} = (W^{\ldm}, \R^{\ldm}_{\Box}, $ $\{R^{\ldm}_{i} | i \in Ag\}, \R^{\ldm}_{Ag}, \R^{\ldm}_{\g}, \R^{\ldm}_{\h},$ $V^{\ldm})$ such that:

\begin{itemize}

\item $W^{\ldm} = W^{\pres} \times \mathbb{N}$\footnote{Note that we choose to write each world $(w,j) \in W^{\ldm}$ as $w^{j}$ to simplify notation. Moreover, we write $\phi \in w^{j}$ to mean that the formula $\phi$ is in the $\ldm$-MCS $w$ associated with $j$.};

\item $\R^{\ldm}_{\Box}w^{j}u^{j}$ iff (i) $\R^{\pres}_{\Box}wu$ and $j = 0$, or (ii) $\R^{\fut}_{\Box}wu$ and $j > 0$;

\item $\R^{\ldm}_{i}w^{j}u^{j}$ iff (i) $\R^{\pres}_{i}wu$ and $j = 0$, or (ii) $\R^{\fut}_{i}wu$ and $j > 0$;

\item $\R^{\ldm}_{Ag}(w^{j}) = \bigcap_{1 \leq i \leq n} \R^{\ldm}_{i}(w^{j})$;

\item $\R^{\ldm}_{\g} = \{(w^{j},w^{k})| w^{j}, w^{k} \in W^{\ldm} \text{ and } j < k \}$;

\item $\R^{\ldm}_{\h} = \{(u^{i},w^{i})| (w^{i},u^{i}) \in \R^{\ldm}_{\g}\}$;

\item $V^{\ldm}(p) = \{w^{j} \in W^{\ldm}| w \in V^{\pres}(p)\}$.

\end{itemize}

\end{definition}

\begin{lemma}\label{Relate_only_same_level} For all $\alpha \in \{\Box, Ag\} \cup Ag$, if $\R^{\ldm}_{\alpha}w^{j}u^{k}$ for $j,k \in \mathbb{N}$, then $j = k$.
\end{lemma}

\begin{proof} Follows by definition of the canonical $\tstit$ model.

\end{proof}

\begin{lemma}\label{AG_Rel_Same_All_Levels} For all $j \in \mathbb{N}$ with $k \geq 1$, $(w^{j},u^{j}) \in \R^{\ldm}_{Ag}$ iff $(w^{j+k},u^{j+k}) \in \R^{\ldm}_{Ag}$.
\end{lemma}

\begin{proof} This follows from the fact that $u^{0} \in \R^{\ldm}_{Ag}(w^{0})$ iff $u \in \bigcap_{i \in Ag} \R^{\pres}_{i}(w)$ iff $u \in \R^{\fut}_{i}(w)$ for each $i \in Ag$ iff $u \in \bigcap_{i \in Ag} \R^{\fut}_{i}(w)$ iff $u^{k} \in \bigcap_{i \in Ag} \R^{\ldm}_{i}(w^{k})$ for any $k > 0$.
\end{proof}

\begin{lemma}[\cite{BlaRijVen01}]\label{Relations_Diamonds} (i) For all $\mathsf{x} \in\{\pres, \fut, \ldm\}$, $\R^{\mathsf{x}}_{\Box}wu$ iff for all $\phi$, if $\phi \in u$, then $\Diamond \phi \in w$. (ii) For all $\mathsf{x} \in\{\pres, \fut, \ldm\}$, $\R^{\mathsf{x}}_{i}wu$ iff for all $\phi$, if $\phi \in u$, then $\lb i \rb \phi \in w$.

\end{lemma}

%\begin{proof} See \cite{BlaRijVen01} for details.
%\end{proof}

\begin{lemma}[Existence Lemma \cite{BlaRijVen01}]\label{Existence_Lemma} (i) For any world $w^{j} \in W^{\ldm}$, if $\Diamond \phi \in w^{j}$, then there exists a world $u^{j} \in W^{\ldm}$ such that $\R^{\ldm}_{\Box}w^{j}u^{j}$ and $\phi \in u^{j}$. (ii) For any world $w^{j} \in W^{\ldm}$, if $\lb i \rb \phi \in w^{j}$, then there exists a world $u^{j} \in W^{\ldm}$ such that $\R^{\ldm}_{i}w^{j}u^{j}$ and $\phi \in u^{j}$.

\end{lemma}

%\begin{proof} See \cite{BlaRijVen01} for details.
%\end{proof}

\begin{lemma}\label{Canonical_is_TTKSTIT_Model}
The Canonical Model is a temporal Kripke STIT model.
\end{lemma}

\begin{proof} We prove that $M^{\ldm}$ has all the properties of a $\tstit$ model:

\begin{itemize}

\item By lemma \ref{Lindenbaum}, the $\ldm$ consistent set $\{p\}$ can be extended to a $\ldm$-MCS, and therefore $W^{\pres}$ is non-empty. Since $\mathbb{N}$ is non-empty as well, $W^{\pres} \times \mathbb{N} = W^{\ldm}$ is a non-empty set of worlds.

\item We argue that $\R^{\ldm}_{\Box}$ is an equivalence relation between worlds of $W^{\ldm}$, and omit the arguments for $\R^{\ldm}_{i}$ and $\R^{\ldm}_{Ag}$, which are similar. Suppose that $w^{j} \in W^{\ldm}$. We have two cases to consider: (i) $j = 0$, and (ii) $j > 0$. \textbf{(i)} Standard canonical model arguments apply and $\R^{\ldm}_{\Box}$ is an equivalence relation between all worlds of the form $w^{0} \in W^{\ldm}$ (See \cite{BlaRijVen01} for details). \textbf{(ii)} If we fix a $j >0$, then $\R^{\ldm}_{\Box}$ will be an equivalence relation for all worlds of the form $w^{j} \in W^{\ldm}$ since the intersection of equivalence relations produces another equivalence relation. Last, since $\R^{\ldm}_{\Box}$ is an equivalence relation for each fixed $j \in \mathbb{N}$, and because each $W^{\pres} \times \{j\} \subset W^{\ldm}$ is disjoint from each $W^{\pres} \times \{j'\} \subset W^{\ldm}$ for $j \neq j'$, we know that the union all such equivalence relations will be an equivalence relation.

\item[{\rm \textbf{(C1)}}] Let $i$ be in $Ag$ and assume that $(w^{j},u^{j}) \in \R^{\ldm}_{i}$. We split the proof into two cases: (i) $j = 0$, or (ii) $j > 0$. \textbf{(i)} Assume that $\Box \phi \in w^{0}$. Since $w$ is a $\ldm$-MCS, it contains the axiom $\Box \phi \rightarrow [i] \phi$, and so, $[i] \phi \in w$ as well. Since $(w,u) \in \R^{\pres}_{i}$ (because $j = 0$), we know that $\phi \in u$ by the definition of the relation; therefore, $(w,u) \in \R^{\pres}_{\Box}$, which implies that $(w^{0},u^{0}) \in \R^{\ldm}_{\Box}$ by definition. \textbf{(ii)} The assumption that $j >0$ implies that $u \in \R^{\fut}_{i} (w) $ $= \bigcap_{i \in Ag} \R^{\pres}_{i}(w) = \R^{\fut}_{\Box}(w)$ by definition, which implies that $(w^{j},u^{j}) \in \R^{\ldm}_{\Box}$.

\item[{\rm \textbf{(C2)}}] Let $u^{j}_{1}, ..., u^{j}_{n} \in W^{\ldm}$ and assume that $\R^{\ldm}_{\Box}u^{j}_{i}u^{j}_{k}$ for all $i,k \in \{1,...,n\}$. We split the proof into two cases: (i) $j = 0$, or (ii) $j > 0$. \textbf{(i)} We want to show that there exists a world $w^{j} \in W^{\ldm}$ such that $w^{j} \in \bigcap_{1 \leq i \leq n} \R^{\ldm}_{i}(u^{j}_{i})$. Let $\hat{w}^{j} = \bigcup_{1 \leq i \leq n} \{\phi | [i] \phi \in u^{j}_{i} \}$. Suppose that $\hat{w}^{j}$ is inconsistent to derive a contradiction. Then, there are $\psi_{1}$,...,$\psi_{k}$ such that $\vdash_{\ldm} \bigwedge_{1 \leq l \leq k} \psi_{i} \rightarrow \bot$. For each $i \in Ag$, we define $\Phi_{i} = \{\psi_{l}|[i] \psi_{l} \in u^{j}_{i}\} \subseteq \{\psi_{1},...,\psi_{k}\}$. Observe that for each $i \in Ag$, $[i] \bigwedge \Phi_{i} \in u^{j}_{i}$ because $\bigwedge [i] \Phi_{i} \in u^{j}_{i}$ and $\vdash_{\ldm} \bigwedge [i] \Phi_{i} \rightarrow [i] \bigwedge \Phi_{i}$. Since by assumption $\R^{\ldm}_{\Box}u^{j}_{i}u^{j}_{k}$ for all $i,k \in \{1,...,n\}$, this means that for any $u^{j}_{m}$ we pick (with $1 \leq m \leq n$), $\Diamond [i] \bigwedge \Phi_{i} \in u^{j}_{m}$ for each $i \in Ag$ by lemma \ref{Relations_Diamonds}; hence, $\bigwedge_{i \in Ag} \Diamond [i] \bigwedge \Phi_{i} \in u^{j}_{m}$. By the $\ioa$ axiom, this implies that $\Diamond \bigwedge_{i \in Ag} [i] (\bigwedge \Phi_{i}) \in u^{j}_{m}$. By lemma \ref{Existence_Lemma}, there must exist a world $v^{j}$ such that $\R^{\ldm}_{\Box}u^{j}_{m}v^{j}$ and $\bigwedge_{i \in Ag} [i] (\bigwedge \Phi_{i}) \in v^{j}$. But then, since $\vdash_{\ldm} [i] (\bigwedge \Phi_{i}) \rightarrow \bigwedge \Phi_{i}$ by reflexivity, $\vdash_{\ldm} \bigwedge_{i \in Ag} (\bigwedge \Phi_{i}) \leftrightarrow \bigwedge_{1 \leq i \leq k} \psi_{i}$, and $\vdash_{\ldm} \bigwedge_{1 \leq i \leq k} \psi_{i} \rightarrow \bot$, it follows that $\bot \in v^{j}$, which is a contradiction since $v^{j}$ is a $\ldm$-MCS. Therefore, $\hat{w}^{j}$ must be consistent and by lemma \ref{Lindenbaum}, it may be extended to a $\ldm$-MCS $w^{j}$. Since for each $[i] \phi \in u^{j}_{i}$, $\phi \in w^{j}$, we have that $w \in \R^{\pres}_{i}(u_{i})$ for each $i \in Ag$. Hence, $w \in \bigcap_{1 \leq i \leq n} \R^{\pres}_{i}(u_{i})$, and so, $w^{j} \in \bigcap_{1 \leq i \leq n} \R^{\ldm}_{i}(u^{j}_{i})$. \textbf{(ii)} Suppose that $j > 0$, so that $t^{j} \in \R^{\ldm}_{\Box}(s^{j})$ iff $t \in \R^{\fut}_{\Box}(s)$ = $\bigcap_{i \in Ag} \R^{\pres}_{i}(s)$. By assumption then, $u^{j}_{m} \in \bigcap_{i \in Ag} \R^{\pres}_{i}(u^{j}_{k}) = \R^{\fut}_{i}(u^{j}_{k})$ for all $k,m \in \{1,...,n\}$ and each $i \in Ag$. Hence, $u^{j}_{m} \in \bigcap_{i \in Ag} \R^{\fut}_{i}(u^{j}_{k})$ for all $k,m \in \{1,...,n\}$. If we therefore pick any $u^{j}_{k}$, it follows that $u^{j}_{k} \in \bigcap_{i \in Ag}\R^{\fut}_{i} (u^{j}_{i})$, meaning that the intersection $\bigcap_{1 \leq i \leq n} \R^{\ldm}_{i}(u^{j}_{i})$ is non-empty.

\item[{\rm \textbf{(C3)}}] Follows by definition.

\item $\R^{\ldm}_{\g}$ is a transitive and serial by definition, and $\R^{\ldm}_{\h}$ is the converse of $\R^{\ldm}_{\g}$ by definition as well.

%\begin{itemize}

\item[{\rm \textbf{(C4)}}] For all $u^{j}, u^{k}, u^{l} \in W^{\ldm}$, suppose that $\R^{\ldm}_{\g}u^{j}u^{k}$ and $\R^{\ldm}_{\g}u^{j}u^{l}$.  Then, $j < k$ and $j < l$, and since $\mathbb{N}$ is linearly ordered, we have that $k < l$, $k = l$, or $k> l$, implying that $\R^{\ldm}_{\g}u^{k}u^{l}$, $u^{k} = u^{l}$, or $\R^{\ldm}_{\g}u^{l}u^{k}$.

\item[{\rm \textbf{(C5)}}] Similar to previous case.

\item[{\rm \textbf{(C6)}}] Suppose that $(u^{j},v^{j+k}) \in \R^{\ldm}_{\g} \circ \R^{\ldm}_{\Box}$ with $k \geq 1$. By definition of $\R^{\ldm}_{\g}$, $u^{j+k}$ is the only element in $\R^{\ldm}_{\g}(u^{j})$ associated with $j+k$, and so, $(u^{j+k},v^{j+k}) \in \R^{\ldm}_{\Box}$ (By lemma \ref{Relate_only_same_level} no other $u^{j+k'}$ with $k' \neq k$ can relate to $v^{j+k}$ in $\R^{\ldm}_{\Box}$.). Since $k \geq 1$, $v^{j+k} \in \R^{\ldm}_{\Box}(u^{j+k})$ iff $v \in \R^{\fut}_{\Box}(u) = \bigcap_{i \in Ag} \R^{\pres}_{i}(u)$ iff $v^{0} \in \R^{\ldm}_{Ag}(u^{0})$. By lemma \ref{AG_Rel_Same_All_Levels}, $(u^{j},v^{j}) \in \R^{\ldm}_{Ag}$. This implies that, and since $(v^{j},v^{j+k}) \in \R^{\ldm}_{\g}$ by definition, we have that $(u^{j},v^{j+k}) \in \R^{\ldm}_{Ag} \circ \R^{\ldm}_{\g}$.

\item[{\rm \textbf{(C7)}}] Follows from the definition of the $\R^{\ldm}_{\g}$ relation.

%For all $w \in W$, if $u \in \R_{\Box}(w)$, then $v \not\in \R_{\g}(w)$;

%\end{itemize}

\item Last, it is easy to see that the valuation function $V^{\ldm}$ is indeed a valuation function.

\end{itemize}

\end{proof}

\begin{lemma}[Truth-Lemma]\label{Truth_Lemma} For any formula $\phi$, $M^{\ldm}, w^{0} \models \phi$ iff $\phi \in w^{0}$.

\end{lemma}

\begin{proof} Shown by induction on the complexity of $\phi$ (See \cite{BlaRijVen01}).

\end{proof}

\section{$\xstitl$ Derivation of IOA$^x$ Axiom}\label{XSTIT_IOA}

We make use of the system of rules $\ioax$, to derive the $\xstit$ IOA axiom in $\xstitl$. \\

%\begin{center}
\begin{small}
%\vspace{-0.6cm}
%\begin{figure}\label{ioaxderivation}
\begin{tabular}{@{\hskip -4em} c}
\AxiomC{$R_{\Box}w_{1}w_{2}, R_{\Box}w_{1}w_{3}, R_{\Box}w_{1}w_{4}, R_{A}w_{4}w_{5}, R_{A}w_{2}w_{5}, w_{2}: \langle A \rangle^x \overline{\phi}, w_{3}: \langle B \rangle^x \overline{\psi}, ...\ \ $% w_{1}: \Diamond ([A]^x\phi\land [B]^x\psi)
$ w_{5}: \phi, w_{5}: \overline{\phi}$}
\UnaryInfC{$R_{\Box}w_{1}w_{2}, R_{\Box}w_{1}w_{3}, R_{\Box}w_{1}w_{4}, R_{A}w_{4}w_{5}, R_{A}w_{2}w_{5}, w_{2}: \langle A \rangle^x \overline{\phi}, w_{3}: \langle B \rangle^x \overline{\psi}, ...\ \ $% w_{1}: \Diamond ([A]^x\phi\land [B]^x\psi)
$ w_{5}: \phi$}
\RightLabel{($\mathsf{IOA-U_{1}})$}
\UnaryInfC{$R_{\Box}w_{1}w_{2}, R_{\Box}w_{1}w_{3}, R_{\Box}w_{1}w_{4}, R_{A}w_{4}w_{5}, w_{2}: \langle A \rangle^x \overline{\phi}, w_{3}: \langle B \rangle^x \overline{\psi}, ... \ \ $%w_{1}: \Diamond ([A]^x\phi\land [B]^x\psi)
$ w_{5}: \phi$}
\UnaryInfC{$R_{\Box}w_{1}w_{2}, R_{\Box}w_{1}w_{3}, R_{\Box}w_{1}w_{4}, w_{2}: \langle A \rangle^x \overline{\phi}, w_{3}: \langle B \rangle^x \overline{\psi}, ...\ \ $%w_{1}: \Diamond ([A]^x\phi\land [B]^x\psi)
$ w_{4}: [A]^x\phi$}
\UnaryInfC{$D_1$}
\DisplayProof
\end{tabular}

\ \\

\begin{tabular}{@{\hskip -4em} c}
\AxiomC{$R_{\Box}w_{1}w_{2}, R_{\Box}w_{1}w_{3}, R_{\Box}w_{1}w_{4}, R_{B}w_{4}w_{6}, R_{B}w_{3}w_{6}, w_{2}: \langle A \rangle^x \overline{\phi}, w_{3}: \langle B \rangle^x \overline{\psi}, ....\ \ $%w_{1}: \Diamond ([A]^x\phi\land [B]^x\psi)
$ w_{6}: \psi, w_{6} : \overline{\psi}$}

\UnaryInfC{$R_{\Box}w_{1}w_{2}, R_{\Box}w_{1}w_{3}, R_{\Box}w_{1}w_{4}, R_{B}w_{4}w_{6}, R_{B}w_{3}w_{6}, w_{2}: \langle A \rangle^x \overline{\phi}, w_{3}: \langle B \rangle^x \overline{\psi}, ...\ \ $% w_{1}: \Diamond ([A]^x\phi\land [B]^x\psi)
$ w_{6}: \psi$}
\RightLabel{($\mathsf{IOA-U_{2}})$}
\UnaryInfC{$R_{\Box}w_{1}w_{2}, R_{\Box}w_{1}w_{3}, R_{\Box}w_{1}w_{4}, R_{B}w_{4}w_{6}, w_{2}: \langle A \rangle^x \overline{\phi}, w_{3}: \langle B \rangle^x \overline{\psi}, ...\ \ $% w_{1}: \Diamond ([A]^x\phi\land [B]^x\psi)
$ w_{6}: \psi$}
\UnaryInfC{$R_{\Box}w_{1}w_{2}, R_{\Box}w_{1}w_{3}, R_{\Box}w_{1}w_{4}, w_{2}: \langle A \rangle^x \overline{\phi}, w_{3}: \langle B \rangle^x \overline{\psi}, ...\ \ $% w_{1}: \Diamond ([A]^x\phi\land [B]^x\psi)
$ w_{4}: [B]^x\psi$}
\UnaryInfC{$D_2$}
\DisplayProof
\end{tabular}

\ \\

\begin{tabular}{@{\hskip -4em} c}

\AxiomC{$D_1\quad\quad\quad$}%\noLine
%\UnaryInfC{$\vdots\quad\quad$}
\AxiomC{$\quad\quad\quad D_2$}%\noLine
%\UnaryInfC{$\quad\quad\vdots$}

\BinaryInfC{$R_{\Box}w_{1}w_{2}, R_{\Box}w_{1}w_{3}, R_{\Box}w_{1}w_{4}, w_{2}: \langle A \rangle^x \overline{\phi}, w_{3}: \langle B \rangle^x \overline{\psi}, w_{1}: \Diamond ([A]^x\phi\land [B]^x\psi), w_{4}: [A]^x\phi\land [B]^x\psi$}
\UnaryInfC{$R_{\Box}w_{1}w_{2}, R_{\Box}w_{1}w_{3}, R_{\Box}w_{1}w_{4}, w_{2}: \langle A \rangle^x \overline{\phi}, w_{3}: \langle B \rangle^x \overline{\psi}, w_{1}: \Diamond ([A]^x\phi\land [B]^x\psi)$}
\RightLabel{($\mathsf{IOA-E}$)}
\UnaryInfC{$R_{\Box}w_{1}w_{2}, R_{\Box}w_{1}w_{3}, w_{2}: \langle A \rangle^x \overline{\phi}, w_{3}: \langle B \rangle^x \overline{\psi}, w_{1}: \Diamond ([A]^x\phi\land [B]^x\psi)$}
\UnaryInfC{$w_{1}: \Box \langle A \rangle^x \overline{\phi}, w_{1}:\Box \langle B \rangle^x \overline{\psi}, w_{1}: \Diamond ([A]^x\phi\land [B]^x\psi)$}
\UnaryInfC{$w_{1}: \Box \langle A \rangle^x \overline{\phi} \lor \Box \langle B \rangle^x \overline{\psi} \lor \Diamond ([A]^x\phi\land [B]^x\psi)$}
\DisplayProof
\end{tabular}
%\caption{Derivation of the IOA axiom in $\xstitl$.}
%\end{figure}
\end{small}
%\end{center}

\end{document}